\documentclass[11pt,a4paper]{article}

\usepackage{amsmath, latexsym, amsfonts, amssymb, amsthm}
\usepackage{graphicx, color,hyperref,dsfont,epsfig,caption,wrapfig,subfig}
\usepackage{pstricks}
\usepackage{movie15}
\usepackage{geometry}
\hypersetup{
    linktoc=page,
    linkcolor=red,          
    citecolor=blue,        
    filecolor=blue,      
    urlcolor=cyan,
    colorlinks=true           
}
\usepackage[makeroom]{cancel}
\usepackage{cite}
\usepackage[affil-it]{authblk}

\definecolor{mygreen}{RGB}{34,139,34} 
\definecolor{myblue}{RGB}{0,0,205}
\definecolor{myred}{RGB}{178,34,34}
\definecolor{myorange}{RGB}{255,127,80}
\definecolor{mylilas}{RGB}{170,55,241}

\numberwithin{equation}{section}
\newtheorem{theo}{Theorem}[section]
\newtheorem{lem}[theo]{Lemma}
\newtheorem{prop}[theo]{Proposition}

\newtheorem{rem}[theo]{Remark}
\newtheorem{definition}[theo]{Definition}

\theoremstyle{definition}

\newcommand{\E}{\mathbb{E}}
\newcommand{\intc}{\int_{\mathbb{C}}}
\newcommand{\intr}{\int_{\mathbb{R}}}
\newcommand{\hg}{\hat{g}}
\renewcommand{\lg}{\langle}
\newcommand{\rg}{\rangle}
\newcommand{\alp}{\boldsymbol{\alpha}}
\newcommand{\x}{\mathbf{x}}
\newcommand{\n}{\mathbf{n}}
\newcommand{\y}{\mathbf{y}}
\newcommand{\z}{\mathbf{z}}
\newcommand{\q}{\mathbf{q}}

\numberwithin{equation}{section}
\hypersetup{
    colorlinks,
    linkcolor={red!50!black},
    citecolor={blue!50!black},
    urlcolor={blue!80!black}
}

\title{\textsc{Higher order BPZ equations for Liouville Conformal Field Theory}}
\author{Tunan Zhu\footnote{This work was supported by grants from R\'egion Ile-de-France.}}
\affil{D\'epartement de Math\'ematiques et Applications,\\ \'Ecole normale sup\'erieure de Paris.}

\begin{document}

\maketitle

\begin{abstract}
Inspired by some intrinsic relations between Coulomb gas integrals and Gaussian multiplicative chaos, this article introduces a general mechanism to prove BPZ equations of order $(r,1)$ and $(1,r)$ in the setting of probabilistic Liouville conformal field theory, a family of conformal field theory which depends on a parameter $\gamma \in (0,2)$. The method consists in regrouping singularities on the degenerate insertion, and transforming the proof into an algebraic problem. With this method we show that BPZ equations hold on the sphere for the parameter $\gamma \in [\sqrt{2},2)$ in the case $(r,1)$ and for $\gamma \in (0,2)$ in the case $(1,r)$. The same technique applies to the boundary Liouville field theory when the bulk cosmological constant $\mu_{\text{bulk}} = 0$, where we prove BPZ equations of order $(r,1)$ and $(1,r)$ for $\gamma \in (0,2)$.  \\

\end{abstract}

\tableofcontents

\section{Introduction}

Liouville conformal field theory (LCFT) falls within the general framework of conformal field theory (CFT). One of the main goals of the theory is to characterize the correlation functions, which can be considered as probability amplitudes for some interacting particle system. A direct relevance with probability theory is their conjectured relation to the scaling limit of large planar maps via the so-called KPZ relation \cite{Knizhnik1988}.

The purpose of this paper is to show that certain correlation functions of LCFT satisfy the Belavin-Polyakov-Zamolodchikov (BPZ) equations, which were first proposed in 1984 \cite{belavin1984infinite} in the general context of CFT. The BPZ equations are indexed by two parameters $ (r, s) $, with $ r, s $ positive integers. The equation associated with the parameter $ (r, s) $ is a partial differential equation of order $ rs $ in several complex variables. There is no general combinatorial formula for the BPZ equations of all orders $(r,s)$. Nevertheless, in 1988, Beno\^it and Saint-Aubin (BSA, \cite{Benoit1988}) found an explicit formula for the BPZ equations of order $ (r, 1) $ and $(1,r)$. The approach that the authors employed is based on the theory of representations. Despite the simplicity, this approach lacks rigorous definitions of the objects involved. 

Recently, in a rigorous mathematical framework, a probabilistic approach to LCFT has been proposed in David-Kupiainen-Rhodes-Vargas \cite {David2016}: the authors construct the correlation functions of LCFT on the sphere using Gaussian multiplicative chaos (GMC). The challenge is to show that the probabilistic setting allows to prove the conjectures made in the physics literature. In this direction, the BPZ equations of order $(2,1)$ and $(1,2)$ have been proved in \cite{Kupiainen2015local}, which constitutes an important step in proving the remarkable DOZZ formula \cite{Kupiainen2017DOZZ}, first proposed by Dorn-Otto-Zamolodchikov-Zamolodchikov \cite{Dorn1994, Zamolodchikov1995}.

Other results on GMC are also proved in different geometries based on the BPZ equations, such as the Fyodorov Bouchaud's formula \cite{Remy2017}, the probabilistic distribution of GMC on the unit interval \cite{Remy2018} and exact formulas for the boundary Liouville structure constants \cite{Fateev2000, Ponsot2002} in an upcoming work. All these series of projects prove the BPZ equations of order $(2,1)$ and $(1,2)$ in a different setting and use this to deduce non trivial shift equations of the object in question, which corresponds to the conformal bootstrap method in physics. 

As a matter of fact, higher order BPZ equations can also be used to deduce exact formulas of certain correlation functions of LFCT, such as the integral forms introduced by Fateev-Litvinov \cite{Fateev2007coulomb, Fateev2008multipoint}. The idea is to first show that the solution space of higher order BPZ equations is of dimension $1$ using monodromy arguments \cite{Dotsenko1984conformal, Dotsenko1985four-point}. It is not hard to verify that the integral forms of Fateev-Litvinov satisfy higher order BPZ equations using its relation with Coulomb gas integrals and analycity of its parameters, especially analycity in $\gamma$. The hard part is to show higher order BPZ equations for Liouville correlation functions where the analycity in $\gamma$ is an open problem. 

In this article, we investigate the intrinsic problem lying in the BPZ equations for LCFT on the sphere and on the unit disk, and we prove that the BSA formula for BPZ equations of order $(r,1)$ and $(1,r)$ holds true for these two cases under some constraints.

\subsection{Basic notions}
The Gaussian free field with vanishing mean on the Riemannian sphere $(\mathbb{C},\hg)$, with $\hg (x) := \frac{4}{(|x|^2+1)^2}$, has covariance given by \cite {David2016}:
\begin{equation}
\E[X(x)X(y)] = \ln \frac{1}{|x-y|}-\frac{1}{4}\left( \ln \hg(x) +\ln \hg(y) \right) +\ln 2 - \frac{1}{2}.
\end{equation}  
Because of the singularity of its covariance, $X$ is not defined pointwise and lives in the space of distributions. We use a regularization for the Gaussian free field $X_{\epsilon} = X*\eta_{\epsilon}$, where the function $\eta_{\epsilon}$ is defined by $\eta_{\epsilon}=\frac{1}{\epsilon^2}\eta(\frac{|x|^2}{\epsilon^2})$, and $\eta \in \mathcal{C}^{\infty}$ is a non-negative smooth function defined on $\mathbb{R}_+$ with compact support in $[\frac{1}{2},1]$ that satisfies $ \pi\int_0^{\infty} \eta (t) dt = 1$. The variance of the regularized field $X_{\epsilon}$ is given by:
\begin{equation}
\E[X_{\epsilon}(x)^2]  = -\frac{1}{2}\ln \hg_{\epsilon}(x)+\ln 2 - \frac{1}{2}, 
\end{equation}
where $\hg_{\epsilon} = \hg * \eta_{\epsilon}$. 

We define the associated GMC measure \cite{Kahane1985} by a standard regularization procedure: for $\gamma \in (0,2)$,
\begin{equation}
e^{\gamma X(x)} \hg(x) d^2x := \lim_{\epsilon \to 0} e^{\gamma X_{\epsilon}(x) - \frac{\gamma^2}{2} \E[X_{\epsilon}(x)^2]} \hg_{\epsilon}(x) d^2x.
\end{equation}
The above convergence is in probability in the weak topology of measures, i.e. for any continuous test function $f:\mathbb{C} \cup \{\infty \} \to \mathbb{R}$, the following limit holds in probability:
\begin{equation}
\intc f(x) e^{\gamma X(x)} \hg(x) d^2x := \lim_{\epsilon \to 0} \intc f(x) e^{\gamma X_{\epsilon}(x) - \frac{\gamma^2}{2} \E[X_{\epsilon}(x)^2]} \hg_{\epsilon}(x) d^2x.
\end{equation}
For an elementary proof of this, see \cite{Berestycki2017}.

Denote $\z = (z_1,\dots, z_N)$, and 
\begin{equation}
U_{N}:=\{(z_1',\dots, z_N') \in \mathbb{C}^N: \forall i \neq j, z_i' \neq z_j' \}.
\end{equation}
We define
\begin{equation}
Q= \frac{\gamma}{2}+\frac{2}{\gamma},
\end{equation}
which is related to the central charge of the LCFT by the formula $c =  1+6Q^2$.
Let us introduce the probabilistic Liouville correlation functions first defined in \cite{David2016}. The definition we give here is coherent with the physics literature and is different from the definition in \cite{David2016} by a multiplicative factor that is of no importance in the setting of this paper.
 
\begin{definition}[Liouville correlations]
For $N \in \mathbb{N}^*$, $\alp \in \mathbb{R}^N$ and $\z \in U_N$, the correlation functions are defined as follows:
	\begin{equation}\label{correlation limit}
	\begin{split}
	&\lg \prod_{l=1}^N V_{\alpha_l}(z_l)\rg := \lim_{\epsilon \to 0} \lg \prod_{l=1}^N V_{\alpha_l,\epsilon}(z_l)\rg\\
	=& \lim_{\epsilon\to 0} 2 e^{(\ln 2 -\frac{1}{2})(\frac{1}{2}\sum_{l=1}^N \alpha_l^2 -\frac{\gamma}{2} \sum_{l=1}^N \alpha_l -\frac{4Q}{\gamma})}\intr e^{-2Qc}\,\E\Big[\prod_{l=1}^N g_{\epsilon}(z_l)^{\Delta_{\alpha_l}} e^{\alpha_l (X_{\epsilon}(z_l)+c) -\frac{\alpha_l^2}{2} \E[X_{\epsilon}(z_l)^2]} \\
	& \quad \times e^{-\mu e^{\gamma c}\intc e^{\gamma X_{\epsilon}(x) -\frac{\gamma^2}{2}\E[X_{\epsilon}(x)^2]} \hg_{\epsilon}(x) d^2x}\Big]dc,
	\end{split}
	\end{equation}
where $\Delta_{\alpha} := \frac{\alpha}{2}(Q-\frac{\alpha}{2})$ is the conformal weight, and $\mu > 0$ is the cosmological constant.

When the Seiberg bounds $\sum_{i=1}^N \alpha_i >2Q$ and $\forall i, \alpha_i < Q$ are satisfied, the limit above exists and converges to the following expression:
\begin{equation}\label{correlation expression}
 Z(\alp)\prod_{1 \le i < j \le N }\frac{1}{|z_i-z_j|^{\alpha_i\alpha_j}} \E\left[\left(\int_{\mathbb{C}} \frac{e^{\gamma X(x)}\hg (x)^{1-\frac{\gamma}{4} \sum_{i=1}^N \alpha_i}d^2x}{\prod_{k=1}^N |x-z_k|^{\gamma\alpha_k} }\right)^{-\frac{\sum_{i=1}^N \alpha_i -2Q}{\gamma}}\right], \\
\end{equation}
where 
\begin{equation}\label{equation Z}
Z(\alp) := 2e^{-\frac{(\ln 2-1/2)}{2}(\sum_{i=1}^N \alpha_i - 2Q)(\sum_{i=1}^N \alpha_i - \frac{4}{\gamma})} \gamma^{-1}\Gamma\left(\frac{\sum_{i=1}^N \alpha_i -2Q}{\gamma},\mu \right).
\end{equation}
\end{definition}

Note that with different conventions, the constant term $Z(\alp)$ can differ, but this will not have any impact on the differential equations. In the notation of $Z(\alp)$ we ignore the dependence on $\gamma$ because the parameter $\gamma$ should be fixed at first to define the background geometry of Liouville fields. The constraint $\sum_{i=1}^N \alpha_i >2Q$ is actually subject to the pole at $0$ of the Gamma function. If we remove the gamma function, the domain of existence can be extended \cite{David2016}:

\begin{lem}[Existence]
$\lg \prod_{l=1}^N V_{\alpha_l}(z_l)\rg/Z(\alp)$ is non trivial if and only if 
\begin{equation}\label{bounds alpha}
(\forall i, \alpha_i < Q) \wedge \left( -\frac{\sum_{i=1}^N \alpha_i -2Q}{\gamma} < \frac{4}{\gamma^2} \wedge \min_i \left\{ \frac{2}{\gamma}(Q-\alpha_i) \right\} \right).
\end{equation}
\end{lem}

This bound is actually the constraint on the moment of total mass of GMC with log singularities. It allows to have positive moments in the expectation. 

Now let us discuss briefly about the regularity of each parameter. It is not hard to show that $\lg \prod_{l=1}^N V_{\alpha_l}(z_l)\rg$ is continuous in $\gamma$, $\alp$ and $\z$ respectively, but we can go further. Correlation functions are actually analytic in $\alp$, proved in \cite{Kupiainen2017DOZZ} for $\alpha_i$ with small imaginary part (the domain of analyticity was then extended in \cite{2018Huang}). The correlation functions are also smooth in $\z$, as proved in a recent work by Oikarinen \cite{Oikarinen2018Smoothness}:

\begin{lem}[Smoothness]\label{lem smoothness}
$U_N \ni \z \mapsto \lg \prod_{l=1}^N V_{\alpha_l}(z_l)\rg$ is $\mathcal{C}^{\infty}$.
\end{lem}

Let us give the definition for boundary Liouville correlations on the unit disk represented by $(\mathbb{H},\hg_{\mathbb{H}})$, where $\hg_{\mathbb{H}}(x) = \frac{4}{|x+i|^4}$ is the background metric. It was studied in \cite{2015Huang} and here we give a version that adds different boundary cosmological constants $\mu_i$
\begin{definition}[Boundary Liouville correlations]
Define the Gaussian free field with Newmann boundary conditions and vanishing mean on the boundary:
\begin{equation}
\E[X_{\mathbb{H}}(x) X_{\mathbb{H}}(y)] = \ln \frac{1}{|x-y||x-\bar{y}|} - \frac{1}{2} \ln \hg_{\mathbb{H}}(x) - \frac{1}{2} \ln \hg_{\mathbb{H}} (y).
\end{equation}
Let $z_1,\dots,z_N \in \mathbb{H}$ pairwise distinct, $-\infty <t_1<\dots <t_M<\infty$ and $\mu_1, \dots, \mu_M >0$ with $\mu_0=\mu_M$ by convention, then the boundary Liouville correlations with $\mu_{\text{bulk}} = 0$ are defined as
\begin{align}
&\lg \prod_{i=1}^NV_{\alpha_i}(z_i)\prod_{j=1}^M B^{\mu_{j-1},\mu_{j}}_{\beta_j}(t_j) \rg_{\mathbb{H}} \nonumber\\
:=&Z(\alp; \boldsymbol{\beta})\prod_{1\le i<i' \le N}(|z_i-z_{i'}||z_i-\overline{z}_{i'}|)^{-\alpha_i\alpha_{i'}}  \prod_{1\le i\le N, 1\le j \le M}|z_i-t_j|^{-\alpha_i\beta_j}\prod_{1\le j<j' \le M}|t_j-t_{j'}|^{-\beta_j\beta_{j'}/2}\nonumber\\
& \times \E\left[ \left(\int_{\mathbb{R}}\frac{1}{\prod_{i=1}^N|u-z_i|^{\gamma\alpha_i} \prod_{j=1}^M|u-t_j|^{\gamma\beta_j/2}}e^{\frac{\gamma}{2}X_{\mathbb{H}}(u)} \hg_{\mathbb{H}}(u)^{\frac{\gamma^2}{8}(p+1)}  d\mu_{\partial}(u) \right)^{-p} \right],
\end{align}
where $p = \frac{2(\sum_{i=1}^N \alpha_i +\frac{1}{2}\sum_{j=1}^M \beta_j-Q)}{\gamma}$ and 
\begin{equation}
\frac{d\mu_{\partial}(u)}{du} = \sum_{j=1}^{M-1} \mu_j \mathbf{1}_{t_j<u < t_{j+1}} + \mu_M \mathbf{1}_{u \notin (t_1, t_M)}.
\end{equation}
\end{definition}
\begin{rem}
$\lg \prod_{i=1}^NV_{\alpha_i}(z_i)\prod_{j=1}^M B^{\mu_{j-1},\mu_{j}}_{\beta_j}(t_j) \rg_{\mathbb{H}}/Z(\alp; \boldsymbol{\beta})$ is well defined if and only if
\begin{equation}
(\forall i, \alpha_i < Q) \wedge ( \forall j, \beta_j < Q) \wedge  \left( -p < \frac{4}{\gamma^2} \wedge \min_{j} \left\{ \frac{2}{\gamma}(Q-\beta_j) \right\} \right).
\end{equation}
\end{rem}

The different values of $\mu_j$ represent boundary cosmological constants on each piece of the boundary. We can send some of the $\mu_j$ to $0$ as long as $d\mu_{\partial}$ is non trivial. The expression of normalization factor $Z(\alp; \boldsymbol{\beta})$ is of no importance since we are only interested in differential equations in $\z$. 

Let us also introduce Coulomb gas integrals that will be useful for proving the BPZ equations. We will explain later in section \ref{section strategy} how these integrals are related to Liouville correlation functions.

\begin{definition}[Coulomb gas integrals]\label{definition Coulomb}
Let $\z \in U_N$, $l \in \mathbb{N}^*$. Define the complex Coulomb gas integrals
\begin{equation}
\mathfrak{C}^{(l)}_{\alp}(\z):= \prod_{1\le i< j \le N} |z_i-z_j|^{-\alpha_i\alpha_j} \int_{\mathbb{C}^{l}}  \prod_{1 \le i \le N , 1\le s \le l}|y_s-z_i|^{-\gamma\alpha_i}\prod_{1\le s<s' \le l }|y_s-y_{s'}|^{-\gamma^2} d^2 \mathbf{y}.
\end{equation}
The integral converges when $\gamma^2< \frac{4}{l}$, $ \forall i\,\,\,\alpha_i < \frac{2}{\gamma}-\frac{(l-1)\gamma}{2}$ and $\sum_{i=1}^N \alpha_i > \frac{2}{\gamma}-\frac{(l-1)\gamma}{2}$.
In the proof of the BPZ equations, we will need real Coulomb gas integrals: for $t_0 < t_1 <\dots < t_N (N \ge 2)$, 
\begin{equation}
C^{(l)}_{\alpha_0,\alp}(t_0,\mathbf{t}):= \prod_{0\le i< j \le N} (t_j-t_i)^{-\frac{\alpha_i\alpha_j}{2}} \int_{t_{N-1}< x_1 < \dots < x_l < t_N} \prod_{0 \le i \le N , 1\le s \le l}(x_s-t_i)^{-\frac{\gamma\alpha_i}{2}} \prod_{1\le s'<s \le l }(x_s-x_{s'})^{-\frac{\gamma^2}{2}} d \mathbf{x},
\end{equation}
where $(-1)^{\alpha}$ depends on the choice of contour and is set to be $e^{i \alpha \pi}$. Especially, when $(\gamma,\alpha_{N-1},\alpha_N) \in (i\mathbb{R}_+)^{3}$, $C^{(l)}_{\alpha_0,\alp}(t_0,\mathbf{t})$ is always well defined.
\end{definition}
\begin{rem}
In section \ref{section proof BPZ}, we will consider $C^{(l)}_{-(r-1)\chi,\alp}(t,\mathbf{t})$ for $\chi=\frac{\gamma}{2} \text{ or } \frac{2}{\gamma}$, and we will work with $\alpha_{N-1},\alpha_N \in i \mathbb{R}_+$ sufficiently large in absolute value to have enough differentiability.
\end{rem}

\subsection{Main results}

\begin{definition}\label{def L}
Denote $\mathcal{L}_-(z;\z)$ the algebra generated by the differential operators $(L_{-n})_{n \ge 1}$ and the identity operator id, where
\begin{equation}
	L_{-1}:=\partial_{z}, \quad L_{-n}:=\sum_{l=1}^N \left(-\frac{1}{(z_l-z)^{n-1}}\partial_{z_l}+\frac{\Delta_{\alpha_l}(n-1)}{(z_l-z)^n}\right) \quad n\ge 2.
	\end{equation}
\end{definition}

In the literature, $V_\alpha$ are called local fields, we also call it an insertion. A field $V_{\alpha}$ is degenerate if $\alpha = -\frac{(r-1)\gamma}{2}-\frac{2(s-1)}{\gamma}$ for $r,s \in \mathbb{N}^*$, in this case we call it a $(r,s)-$degenerate insertion. When there is a degenerate field, physicists have predicted that correlation functions of CFT satisfy certain partial differential equations with highest order $\partial^{rs}_z$ known as the BPZ equations. Although it is theoretically possible to construct this differential equation from operators of $\mathcal{L}_-(z;\z)$, there is no general formula to achieve this. Only in the case when $r=1$ or $s=1$, Beno\^it and Saint-Aubin \cite{Benoit1988} found an explicit and compact formula:

\begin{theo}\label{theo BSA}
Let $r\ge 2$ an integer and
\begin{equation}
\chi = \frac{\gamma}{2} \text{ or } \frac{2}{\gamma}.
\end{equation}
The BPZ equations of order $r$ hold true for $\gamma \in ( \sqrt{\frac{2(r-2)}{r-1}},2)$ when $\chi = \frac{\gamma}{2}$ and for $\gamma \in (0,2)$ when $\chi = \frac{2}{\gamma}$:
\begin{equation}
\mathcal{D}_{r} \lg V_{-(r-1)\chi}(z) \prod_{l=1}^{N}V_{\alpha_l}(z_l)\rg=0,
\end{equation}
where the differential operator $\mathcal{D}_{r} $ is given by the Beno\^it and Saint-Aubin's formula:
	\begin{equation}\label{equation Dr}
	\mathcal{D}_{r}=\sum_{k=1}^r \sum_{
	\substack{(n_1,\,\dots,\,n_k)\in (\mathbb{N^*})^k \\
		n_1+\dots+n_k=r
	}} \frac{(\chi^2)^{r-k}}{\prod_{j=1}^{k-1}(\sum_{i=1}^j n_i)(\sum_{i=j+1}^k n_i)} L_{-n_1}\dots L_{-n_k},
	\end{equation}
with $L_{-n}$ defined in definition \ref{def L}. 
\end{theo}

\begin{rem}
In this article we use BPZ of order $r$ to represent the two cases of order $(r,1)$ and $(1,r)$. We remark that BPZ equations of order 2 has been well investigated by Kupiainen-Rhodes-Vargas in \cite{Kupiainen2015local}. The theorem above generalizes their result to all $r \ge 2$, with a constraint on $\gamma$ when $\chi = \frac{\gamma}{2}$. With some slight efforts, the order $3$ BPZ equations can be proved for all $\gamma \in (0,2)$. This will be shown in section \ref{section illustration}.
\end{rem}

\begin{rem}
The constraint on $\gamma$ when $\chi = \frac{\gamma}{2}$ is a purely technical condition and is only required by Proposition \ref{prop R to 0}. In particular, we have BPZ equations of all orders for $\gamma \in [\sqrt{2}, 2)$ when $\chi = \frac{\gamma}{2}$.
\end{rem}

In the boundary LCFT case we can have a boundary degenerate insertion or a bulk degenerate insertion. When it comes to a boundary degenerate insertion $B_{-(r-1)\chi}^{\pm}(t)$ defined as below, we will work with an extended definition where $t$ lives in the upper-half plane:
\begin{definition}
Let $\mu_1 ,\dots \mu_{M} >0$, $-\infty <t_1 <\dots < t_M < \infty$ and $t \in \mathbb{H} \backslash \{z_i, 1\le i \le N\} \cup \mathbb{R}$. We define the extended correlation function $\lg B^{+}_{-(r-1)\chi}(t) \prod_{i=1}^NV_{\alpha_i}(z_i)\prod_{j=1}^M B_{\beta_j}^{\mu_{j-1},\mu_j}(t_j) \rg_{\mathbb{H}}$ by
\begin{small}
\begin{align}\label{equation boundary degenerate}
&\prod_{i=1}^N((z_i-t)(\overline{z}_i-t))^{\frac{(r-1)\chi\alpha_i}{2}} \prod_{j=1}^M (t_i-t)^{\frac{(r-1)\chi\beta_j}{2}}\prod_{1\le i<i' \le N}(|z_i-z_{i'}||z_i-\overline{z}_{i'}|)^{-\alpha_i\alpha_{i'}} \prod_{1\le i \le N, 1\le j \le M}|z_i-t_j|^{-\alpha_i\beta_j}\nonumber\\
&\prod_{1\le j<j'\le M}|t_j-t_{j'}|^{-\beta_j\beta_{j'}/2}   
\E\left[ \left(\int_{\mathbb{R}}\frac{(t-u)^{\frac{(r-1)\gamma\chi}{2}}}{|u-z_i|^{\gamma\alpha_i}|u-t_j|^{\gamma\beta_j/2}}e^{\frac{\gamma}{2}X_{\mathbb{H}}(u)} \hg_{\mathbb{H}}(u)^{\frac{\gamma^2}{8}(p-\frac{(r-1)\chi}{\gamma}+1)}  d\mu_{\partial}(u) \right)^{-p+\frac{(r-1)\chi}{\gamma}} \right].
\end{align}
\end{small}
Similarly, we define $\lg B^{-}_{-(r-1)\chi}(t) \prod_{i=1}^NV_{\alpha_i}(z_i)\prod_{j=1}^M B_{\beta_j}^{\mu_{j-1},\mu_j}(t_j) \rg_{\mathbb{H}}$ by replacing the term $(t-u)^{\frac{(r-1)\gamma\chi}{2}}$ in the above integral by $(u-t)^{\frac{(r-1)\gamma\chi}{2}}$.
\end{definition}

\begin{rem}
When $t\in \mathbb{R}$ satisfies $t_{i_0} < t < t_{i_0+1}$ for certain $1\le i_0 \le M-1$, we have 
\begin{align*}
\lg B^{\pm}_{-(r-1)\chi}(t) \prod_{i=1}^NV_{\alpha_i}(z_i)\prod_{j=1}^M B_{\beta_j}^{\mu_{j-1},\mu_j}(t_j) \rg_{\mathbb{H}} = \lg B^{\mu_{i_0},\mu_{i_0}e^{\pm i \pi \frac{\gamma\chi}{2}}}_{-(r-1)\chi}(t) \prod_{i=1}^NV_{\alpha_i}(z_i)\prod_{j=1}^M B_{\beta_j}^{\mu_{j-1},\mu_j}(t_j) \rg_{\mathbb{H}}.
\end{align*} 
We also have similar results when $t < t_1$ and $t > t_M$. This explains the reason that we call it an extended correlation function.
\end{rem}

Now we state the BPZ equations for boundary LCFT, where we can prove the result without constraint on $\gamma$.

\begin{theo}\label{theo boundary BPZ}
Let $r\ge 2$ an integer and $\chi = \frac{\gamma}{2} \text{ or } \frac{2}{\gamma}$. Let $\mu_1, \dots, \mu_{M} > 0$, $t_1 < \dots < t_M$ and $t \in \mathbb{H} \backslash \{z_i, 1 \le i \le N\}$. The BPZ equations of order $r$ for a boundary degenerate insertion hold true for $\gamma \in (0,2)$:
\begin{equation}
\mathcal{D}_{r}^{\mathbb{H}} \lg B^{\pm}_{-(r-1)\chi}(t) \prod_{i=1}^NV_{\alpha_i}(z_i)\prod_{j=1}^M B_{\beta_j}^{\mu_{j-1},\mu_j}(t_j) \rg_{\mathbb{H}}=0,
\end{equation}
where the expression of the differential operator $\mathcal{D}_{r}^{\mathbb{H}} $ is given by \eqref{equation Dr}, where we replace the operators $L_{-n}$ by $L_{-n}^{\mathbb{H}}$ defined as $L_{-1}^{\mathbb{H}}:=\partial_{t}$, and for $n \ge 2$:
\begin{align}
L_{-n}^{\mathbb{H}}:=&\sum_{l=1}^N \left(-\frac{1}{(z_l-t)^{n-1}}\partial_{z_l}-\frac{1}{(\overline{z}_l-t)^{n-1}}\partial_{\overline{z}_l}+\frac{\Delta_{\alpha_l}(n-1)}{(z_l-t)^n}+\frac{\Delta_{\alpha_l}(n-1)}{(\overline{z}_l-t)^n}\right) \nonumber\\
& +\sum_{l=1}^M \left(-\frac{1}{(t_l-t)^{n-1}}\partial_{t_l}+\frac{\Delta_{\beta_l}(n-1)}{(t_l-t)^n}\right).
\end{align}

The BPZ equations of order $r$ also hold true when we insert a bulk degenerate insertion: for $\gamma \in (0,2)$,
\begin{equation}
\mathcal{D}_{r}^{\mathbb{H},z} \lg V_{-(r-1)\chi}(z) \prod_{i=1}^NV_{\alpha_i}(z_i)\prod_{j=1}^M B_{\beta_j}^{\mu_{j-1},\mu_j}(t_j) \rg_{\mathbb{H}}=0,
\end{equation}
where $\mathcal{D}_{r}^{\mathbb{H},z}$ is defined by the expression \eqref{equation Dr} where we replace $L_{-n}$ by $L_{-n}^{\mathbb{H},z}$ defined as $L_{-1}^{\mathbb{H},z}:=\partial_{z}$, and for $n \ge 2$,
\begin{align}
L_{-n}^{\mathbb{H},z}:=&\sum_{l=1}^N \left(-\frac{1}{(z_l-z)^{n-1}}\partial_{z_l}-\frac{1}{(\overline{z}_l-z)^{n-1}}\partial_{\overline{z}_l}+\frac{\Delta_{\alpha_l}(n-1)}{(z_l-z)^n}+\frac{\Delta_{\alpha_l}(n-1)}{(\overline{z}_l-z)^n}\right) \nonumber\\
& -\frac{1}{(\overline{z}-z)^{n-1}}\partial_{\overline{z}}+\frac{\Delta_{-(r-1)\chi}(n-1)}{(\overline{z}-z)^n}+\sum_{l=1}^M \left(-\frac{1}{(t_l-z)^{n-1}}\partial_{t_l}+\frac{\Delta_{\beta_l}(n-1)}{(t_l-z)^n}\right).
\end{align}
\end{theo}

In the boundary LCFT case, the proof of BPZ equations is very similar to the sphere case but there is no more technical difficulties, see section \ref{section boundary BPZ}.

\subsection{Strategy of the proof}\label{section strategy}

Let us start by explaining the motivation of introducing Coulomb gas integrals and how it relates to Liouville correlations in a natural way. We consider $-\frac{\sum_{i=1}^N \alpha_i -2Q}{\gamma}=n \in \mathbb{N}^*$. Under this condition, the moment of Liouville correlations can be expanded by Fubini (the rigorous way is to take a regularization for $X$):

\begin{align*}
&\E\left[\left(\int_{\mathbb{C}} \frac{e^{\gamma X(x)}\hg (x)^{1-\frac{\gamma}{4} \sum_{i=1}^N \alpha_i}d^2x}{\prod_{k=1}^N |x-z_k|^{\gamma\alpha_k} }\right)^{n}
\right] =\int_{\mathbb{C}^n} \prod_{j=1}^n \frac{\hg (x_j)^{1-\frac{\gamma}{4} \sum_{i=1}^N \alpha_i}}{\prod_{k=1}^N |x_j-z_k|^{\gamma\alpha_k} } \prod_{i<j} e^{\gamma^2 \E[X(x_i)X(x_j)]} d^2\x\\
&=e^{\frac{n(n-1)\gamma^2}{2}(\ln 2 -\frac{1}{2})}\int_{\mathbb{C}^n} \prod_{j=1}^n \frac{\hg (x_j)^{1-\frac{\gamma}{4} \sum_{i=1}^N \alpha_i - \frac{(n-1)\gamma^2}{4}}}{\prod_{k=1}^N |x_j-z_k|^{\gamma\alpha_k} } \prod_{i<j} \frac{1}{|x_i-x_j|^{\gamma^2}} d^2\x\\
&= e^{\frac{n(n-1)\gamma^2}{2}(\ln 2 -\frac{1}{2})}\int_{\mathbb{C}^n} \prod_{j=1}^n \frac{1}{\prod_{k=1}^N |x_j-z_k|^{\gamma\alpha_k} } \prod_{i<j} \frac{1}{|x_i-x_j|^{\gamma^2}} d^2\x.
\end{align*}
Together with the expression of $Z(\alp)$ \eqref{equation Z}, we deduce that when $-\frac{\sum_{i=1}^N \alpha_i -2Q}{\gamma}=n$,
\begin{equation}\label{equation corr coulomb}
\lg \prod_{l=1}^N V_{\alpha_l}(z_l)\rg = 2 \mathfrak{C}^{(n)}_{\alp}(\z).
\end{equation}
It is explained in \cite{Vargas2017} how physicists use this relation to predict exact formulas on correlation functions of LCFT.

Now we explain the strategy. Consider the Liouville correlation function on the sphere with a degenerate insertion: $\lg V_{-(r-1)\chi}(z) \prod_{l=1}^{N}V_{\alpha_l}(z_l)\rg$. By taking successive derivatives following the operators $L_{-n}$ (see section \ref{section derivative}), we will have integrals that have singularities at $z$ and $z_l$. Using integration by parts and some identities we can regroup all the singularities on $z$. By doing so we observe some repeating terms $P_{\n}Q_{\q}$ (Definition \ref{definition Q}). This allows us to transform the proof of the BPZ equation into an algebraic problem where we search to cancel the coefficients before each $P_{\n}Q_{\q}$.  

On the other hand, we can prove directly that real Coulomb gas integrals satisfy BPZ equations. This is based on the fact that the integrand satisfies BPZ equations (see section \ref{section integrand}). Furthermore, real Coulomb gas integrals have the same algebraic development into $P_{\n}Q_{\q}$ as Liouville correlations, but with a different definition for the quantities $P_{\n}$ and $Q_{\q}$. This is not a surprising fact from the previous explanation of their relations. A study of linear independence of this family allows to show that all the coefficients are actually zero, which means that Liouville correlations satisfy BPZ equations. Remark that we use real Coulomb gas integrals instead of complex ones in order to avoid the problem of integrating against the singularities.

For the organization of this paper, we will present a detailed proof for Theorem \ref{theo BSA} in section \ref{section BPZ}, and in \ref{section boundary BPZ} we give the proof of Theorem \ref{theo boundary BPZ}. Section \ref{section integrand} provides an original and elementary proof showing that the integrand of Coulomb gas integrals satisfy BPZ equations, which implies as a consequence that real Coulomb gas integrals also satisfy the BPZ equations. 
\vspace{5mm}

\textbf{Acknolwedgements:} I would first like to thank R\'emy Rhodes and Vincent Vargas for making me discover LCFT. I also very warmly thank Yichao Huang, Joona Oikarinen, Eveliina Peltola, and Guillaume Remy for many fruitful discussions.

\section{Proof of the BPZ equations}\label{section BPZ}

The subsections \ref{section derivative} to \ref{section proof BPZ} are devoted to the proof of the BPZ equations on the sphere. In section \ref{section boundary BPZ} we will see that the BPZ equations for boundary LCFT can be proved in exactly the same manner as the sphere case, but without constraint on $\gamma$ since the technical problem is avoided by taking $\mu_{bulk} = 0$.

\subsection{Derivatives of correlation functions}\label{section derivative}

We shall first understand how to derive the correlation functions. A proof for the derivative rule is recalled in the appendix \ref{section proof derivative}. 

In this subsection, we will consider the correlation functions $\lg \prod_{l=0}^N V_{\alpha_l,\epsilon}(z_l)\rg$ with $z_0 = z$, $\alpha_0 = -(r-1)\chi$. This is to stay consistent in notations with the later proof of the BPZ equations, but all the results in this subsection hold true for general values of $\alpha_0$ and $z_0$.  Let $\theta: \mathbb{R_+} \to [0,1]$ be a smooth function that equals $0$ in $[0,\frac{1}{2}]$ and $1$ in $[1,\infty)$ and define $\theta_{\delta}=\theta (\frac{|\cdot|}{\delta})$ a regularization function. We introduce the notations:
\begin{definition} 
Define for $\delta > 0$, and $(z,\z) \in U_{N+1}$:
\begin{small}
\begin{align}
\lg \prod_{l=0}^N V_{\alpha_l,\epsilon}(z_l)\rg_{\delta}=&2 e^{(\ln 2 -\frac{1}{2})(\frac{1}{2}\sum_{l=0}^N \alpha_l^2 -\frac{\gamma}{2} \sum_{l=0}^N \alpha_l -\frac{4Q}{\gamma})}\intr e^{-2Qc}\,\E\Big[\prod_{l=0}^N g_{\epsilon}(z_l)^{\Delta_{\alpha_l}} e^{\alpha_l (X_{\epsilon}(z_l)+c) -\frac{\alpha_l^2}{2} \E[X_{\epsilon}(z_l)^2]} \nonumber\\
	& \quad \times e^{-\mu e^{\gamma c}\intc \theta(x-z_0) e^{\gamma X_{\epsilon}(x) -\frac{\gamma^2}{2}\E[X_{\epsilon}(x)^2]} \hg_{\epsilon}(x) d^2x}\Big]dc,
\end{align}
\end{small}
where we add a regularization around $z_0$ while integrating the GMC measure compared to the expression of $\lg \prod_{l=0}^N V_{\alpha_l,\epsilon}(z_l)\rg$ defined in \eqref{correlation limit}. 
\end{definition}

We denote $\lg \prod_{l=0}^N V_{\alpha_l}(z_l)\rg_{\delta}$ the limit of $\lg \prod_{l=0}^N V_{\alpha_l,\epsilon}(z_l)\rg_{\delta}$ when $\epsilon$ goes to $0$, which equals
\begin{equation}\label{equation 2.2}
 Z(\alpha_0, \alp)\prod_{0\le i < j \le N }|z_i-z_j|^{-\alpha_i\alpha_j} \E[(\int_{\mathbb{C}} \frac{\theta_{\delta}(x-z_0)e^{\gamma X(x)}\hg (x)^{1-\frac{\gamma}{4} \sum_{i=0}^N \alpha_i}d^2x}{\prod_{k=0}^N |x-z_k|^{\gamma\alpha_k} })^{-\frac{\sum_{i=0}^N \alpha_i-2Q}{\gamma}}]. 
\end{equation}
The notation $\alp$ still stands for $(\alpha_1, \dots, \alpha_N)$.  Note that $\lg \prod_{l=0}^N V_{\alpha_l}(z_l)\rg_{\delta}$ converges in the weak topology to $\lg \prod_{l=0}^N V_{\alpha_l}(z_l)\rg$.

\begin{lem}[Derivative rule]\label{lem derivation G}
For  $(z_0,\z) \in U_{N+1}$ and $0 \le i \le N$,
\begin{align}
\partial_{z_i}\lg \prod_{l=0}^N V_{\alpha_l,\epsilon}(z_l)\rg_{\delta} = &\sum_{\substack{j=0\\j\neq i}}^N \frac{\alpha_i\alpha_j}{2(z_j-z_i)_{\epsilon}}\lg \prod_{l=1}^N V_{\alpha_l,\epsilon}(z_l)\rg_{\delta}-\frac{\mu \gamma \alpha_i}{2} \intc \frac{\theta_{\delta}(y-z_0)}{(y-z_i)_{\epsilon}}\lg V_{\gamma,\epsilon}(y)\prod_{l=0}^N V_{\alpha_l,\epsilon}(z_l)\rg_{\delta}d^2y \nonumber\\
&+\mathbf{1}_{\{i=0\}} \mu \intc \partial_z \theta_{\delta}(y-z_0)\lg V_{\gamma,\epsilon}(y)\prod_{l=0}^N V_{\alpha_l,\epsilon}(z_l)\rg_{\delta} d^2y,
\end{align}
where
\begin{equation}\label{equation parenthesis}
\frac{1}{(z)_{\epsilon}}:=\intc\intc \frac{1}{z-x_1+x_2}\eta_{\epsilon}(x_1)\eta_{\epsilon}(x_2)d^2x_1 d^2x_2.
\end{equation}
\end{lem}
\begin{rem}
The functions $\frac{1}{(z)_{\epsilon}}$ and $\lg \prod_{l=0}^N V_{\alpha_l,\epsilon}(z_l)\rg_{\delta}$ are smooth. The only difference in this derivative rule with \cite{Kupiainen2015local} is that we take the regularization $\theta_{\delta}$.
\end{rem}
Let us explain briefly how to understand this derivative rule from the expression \eqref{equation 2.2}. There are three terms. The first comes from the preceding term $\prod_{i < j }|z_i-z_j|^{-\alpha_i\alpha_j} $ with regularization. The other two terms appear whenever we take derivatives on moment of Gaussian multiplicative chaos. We can consider them as a simple derivative under expectation and then an application of the Girsanov's theorem. Finally we state an identity that will be useful:

\begin{lem}[KPZ Identity]
For $\delta,\epsilon \ge 0$, the integral $\intc  \theta_{\delta}(y-z) \lg V_{\gamma,\epsilon}(y)\prod_{l=0}^N V_{\alpha_l,\epsilon}(z_l)\rg_{\delta} d^2y$ is well defined and
\begin{equation}\label{equation KPZ}
\mu \gamma \intc \theta_{\delta}(y-z) \lg  V_{\gamma,\epsilon}(y)\prod_{l=0}^N V_{\alpha_l,\epsilon}(z_l)\rg_{\delta} d^2y = (\sum_{l=0}^N \alpha_l-2Q)\lg \prod_{l=0}^N V_{\alpha_l,\epsilon}(z_l)\rg_{\delta}.
\end{equation}
\end{lem}
\begin{rem}
When $\delta$ or $\epsilon$ equal $0$, it simply means that there is no regularization. By applying the lemma multiple times, we obtain in particular that for $p \ge 1$, the integral
\begin{align*}
\int_{\mathbb{C}^p}  \prod_{j=1}^p \theta_{\delta}(y_j-z) \lg \prod_{j=1}^p V_{\gamma,\epsilon}(y_j)\prod_{l=0}^N V_{\alpha_l,\epsilon}(z_l)\rg_{\delta} d^2 \y
\end{align*}
is well defined. An important information to extract from this is the integrability at infinity of the above integral. 
\end{rem}
\begin{proof}
For $\epsilon > 0$ and $\delta \ge 0$, by a change of variable $c' = \frac{\ln \mu}{\gamma}+c$, we have
\begin{align*}
\lg \prod_l V_{\alpha_l,\epsilon}(z_l)\rg_{\delta} =& \mu^{-\frac{\sum_l\alpha_l-2Q}{\gamma}}2 e^{(\ln 2 -\frac{1}{2})(\frac{1}{2}\sum_{l=1}^N \alpha_l^2 -\frac{\gamma}{2} \sum_{l=1}^N \alpha_l -\frac{4Q}{\gamma})}\\
&\intr e^{-2Qc'}\E\Big[\prod_{l=1}^N g_{\epsilon}(z_l)^{\Delta_{\alpha_l}} e^{\alpha_l (X_{\epsilon}(z_l)+c') -\frac{\alpha_l^2}{2} \E[X_{\epsilon}(z_l)^2]} e^{- e^{\gamma c'}\intc e^{\gamma X_{\epsilon}(x) -\frac{\gamma^2}{2}\E[X_{\epsilon}(x)^2]} \hg_{\epsilon}(x) d^2x}\Big]dc'.
\end{align*}
We obtain the lemma by taking the derivative with respect to $\mu$ on both sides. The case $\epsilon = 0$ can be obtained by sending $\epsilon \to 0$.
\end{proof}

In \cite{Kupiainen2015local}, the authors proved that the functions $$y \to \sup_{\epsilon}\lg V_{\gamma,\epsilon}(y)\prod_{l=1}^N V_{\alpha_l,\epsilon}(z_l)\rg \text{ and } (x,y)\to\sup_{\epsilon}\lg V_{\gamma,\epsilon}(x)V_{\gamma,\epsilon}(y)\prod_{l=1}^N V_{\alpha_l,\epsilon}(z_l)\rg$$ are integrable. The results generalize easily to the case with $\theta_{\delta}$. This fact will be useful later to justify the convergences.

\subsection{Repeating patterns}

Let us introduce some notations for the terms that will play a central role in the proof of BPZ equations.

\begin{definition}\label{definition Q}
For $n\in \mathbb{N}^*$, we define
\begin{equation}
P_n(z,\mathbf{z}):=\sum_{l=1}^N\frac{\alpha_l}{2(z_l-z)^n}
\end{equation}
with $\mathbf{z}=(z_1,\dots,z_N)$. For $\n = (n_1,\dots,n_m)$, note $P_{\n}=\prod_{i=1}^{m} P_{n_i} $. 

Let $p\in \mathbb{N}$, and $\mathbf{q}=(q_1, \dots, q_p) \in (\mathbb{N}^*)^p$, we define
\begin{equation}
Q_{\mathbf{q}}(z,\mathbf{z}):=(\frac{\mu \gamma}{2})^p \int_{\mathbb{C}^p} \prod_{j=1}^p\frac{\theta_{\delta}(y_j-z)}{(y_j-z)^{q_j}}\lg V_{-(r-1)\chi}(z) \prod_{i=1}^p V_{\gamma}(y_i) \prod_{l=1}^N V_{\alpha_l}(z_l)\rg_{\delta} \,d^2\mathbf{y}
\end{equation} 
We also provide an operator $T_k$ on $Q_{\mathbf{q}}$, with $k \in \mathbb{N}^*$:
\begin{equation}
T_k{Q_{\mathbf{q}}}=Q_{q_1, \dots, q_k+1, \dots, q_p}
\end{equation}
\end{definition}

Later we will show that proving the BPZ equations is equivalent to a combinatorial problem in the algebra generated by $P_{\n}$ and $Q_{\q}$. 

\begin{definition}
Denote $\mathfrak{R}_{\delta}$ for an arbitrary term in the functional vectorial space: 
\begin{align}
\text{Vect} \Bigg( (z,\z)\mapsto D\left[ P_{\n}\int_{\mathbb{C}^p}  \frac{ \partial_z\theta_{\delta}(y_1-z)}{(y_1-z)^{q_1}} \prod_{i=2}^{p}\frac{\theta_{\delta}(y_i-z)}{(y_i-z)^{q_i}} \lg V_{-(r-1)\chi }(z) \prod_{i=1}^p V_{\gamma}(y_i) \prod_{l=1}^N V_{\alpha_l}(z_l)\rg_{\delta} \,d^2\mathbf{y}\right],\nonumber\\
D\in \mathcal{L}_-(z;\y), \n \in \mathbb{N}_*^m (m\ge 0), \q \in \mathbb{N}_*^{p}(p \ge 1), \sum_{j=1}^p q_j \le r-1 \Bigg). 
\end{align}
\end{definition}
\begin{rem}
$D\mathfrak{R}_{\delta} = \mathfrak{R}_{\delta}$ for all $D \in \mathcal{L}_-(z;\y)$.
\end{rem}
The reason of introducing $\mathfrak{R}_{\delta}(z,\z)$ is that they appear in the calculus of derivatives as perturbation terms and we want to control these terms. We show that they do not have contribution:

\begin{prop}\label{prop R to 0}
When $\chi = \frac{\gamma}{2}$ and $\gamma \in (\sqrt{\frac{2(r-2)}{r-1}},2)$ or when $\chi = \frac{2}{\gamma}$ and $\gamma \in (0,2)$, $\mathfrak{R}_{\delta}(z,\z)$ converges weakly to $0$ in the sense of distributions as $\delta \to 0$.
\end{prop}
\begin{proof}
We can set $D$ to identity since if we have weak convergence to $0$ when $D = \textnormal{id}$, then applying differential operators from $\mathcal{L}_{-}(z;\z)$ will not affect its weak convergence to $0$. Without loss of generality, we take $P_{\n}=1$. It is easy to see that  $\partial_z\theta_{\delta}$ is supported in $B(0,\delta)\backslash B(0,\frac{\delta}{2})$, and $\| \partial_z \theta_{\delta} \|_{\infty} \le \frac{c}{\delta} $ for a constant $c>0$. Then it suffices to control
\begin{align*}
&\quad \left|\int_{\mathbb{C}^p}  \frac{ \partial_z\theta_{\delta}(y_1-z)}{(y_1-z)^{q_1}} \prod_{i=2}^{p}\frac{\theta_{\delta}(y_i-z)}{(y_i-z)^{q_i}} \lg V_{-(r-1)\chi }(z) \prod_{i=1}^p V_{\gamma}(y_i) \prod_{l=1}^N V_{\alpha_l}(z_l)\rg_{\delta} \,d^2\mathbf{y} \right|\\
&\le  c' \delta^{-1-\sum_{i=1}^p q_i} \int_{B(z,\delta)\backslash B(z,\frac{\delta}{2})}\int_{\mathbb{C}^{p-1}} \prod_{i=2}^{p}\theta_{\delta}(y_i-z) \lg V_{-(r-1)\chi }(z) \prod_{i=1}^p V_{\gamma}(y_i) \prod_{l=1}^N V_{\alpha_l}(z_l)\rg_{\delta} \,d^2\mathbf{y}\\
&\overset{\eqref{equation KPZ}}{\le} c(\alp, \gamma, \mu)\delta^{-r} \int_{B(z,\delta)\backslash B(z,\frac{\delta}{2})}\lg V_{-(r-1)\chi }(z) V_{\gamma}(y_1) \prod_{l=1}^N V_{\alpha_l}(z_l)\rg_{\delta} d^2 y_1 .
\end{align*}
Consider $(z,\z)$ in a compact of $U_{N+1}$. We take $\delta_0 <1 \wedge \min_{i\neq j}|z_i -z_j| \wedge \min_{i}|z_i-z|$, then
\begin{align*}
\lg V_{-(r-1)\chi }(z) V_{\gamma}(y_1) \prod_{l=1}^N V_{\alpha_l}(z_l)\rg_{\delta} &\le |y_1-z|^{(r-1)\gamma\chi} Z(-(r-1)\chi,\gamma,\alp)  \prod_{i=1}^N |z_i-z|^{(r-1)\chi\alpha_i}\prod_{i < j }\frac{1}{|z_i-z_j|^{\alpha_i\alpha_j}}\\
&\quad \sup_{y' \in B(z,\delta_0)} \E[(\int_{B(0,1)^c} \frac{|x-z|^{(r-1)\chi} e^{\gamma X(x)}\hg (x)^{1-\frac{\gamma}{4} \sum_{i=1}^N \alpha_i}d^2x}{|x-y'|^{\gamma^2}\prod_{k=1}^N |x-z_k|^{\gamma\alpha_k} })^{-\frac{\sum_{i=1}^N \alpha_i-2Q}{\gamma}}]\\
&\le c(\alp, \gamma, \mu) |y_1-z|^{(r-1)\gamma\chi}.
\end{align*}
Then we can bound the whole term by $c(\alp, \gamma, \mu)\delta^{(r-1)\gamma\chi+2-r}$, which converges to $0$ when the condition on $\gamma$ is satisfied. 

\end{proof}

\subsection{Recursive formulas}\label{section recursive formulas}

This subsection is devoted to proving a recursive formula that allows to transform the higher BPZ equations into a combinatorial form. The main result is the following proposition.

\begin{prop}\label{prop L_-nQ}
The following relation holds when $n + |\n|+|\q| \le r$:
\begin{small}
\begin{align}
	L_{-n}P_{\n}Q_{\mathbf{q}}=& \Bigg[ \sum_{i} n_i \frac{P_{n_i+n}}{P_{n_i}}-\sum_{i=1}^{n-1}P_iP_{n-i}+ ( (n-1)Q - (r-1)\chi )P_n+2\sum_{i=1}^{n-1}P_iT_{p+1}^{n-i} + ((r-1)\chi-\frac{2(n-1)}{\gamma})T_{p+1}^n\nonumber\\
	&-\sum_{i=1}^{n-1}T_{p+2}^{n-i} T_{p+1}^i+\sum_{j=1}^p \Bigg(-\gamma\sum_{i=1}^{n-1}P_iT_j^{n-i}+(\frac{(n-1)\gamma Q}{2} -\frac{(r-1)\gamma\chi}{2} +q_j)T_j^n\nonumber \\
	&+ \gamma\sum_{i=1}^{n-1}T_{p+1}^{n-i}T_j^i - \frac{\gamma^2}{4} \sum_{j'=1}^p \sum_{i=1}^{n-1}T_{j'}^{n-i}T_{j}^i \Bigg) \Bigg] P_{\n}Q_{\mathbf{q}}+\mathfrak{R}_{\delta}
\end{align}
\end{small}
\end{prop}

The recursive relation seems complicated but we will not use directly this expression, what we need is only the homogeneity of $P_{\n}Q_{\q}$. The proposition shows that $\mathcal{D}_{r} Q_0 $ can be expressed as 
\begin{equation}\label{equation lambda}
\mathcal{D}_{r} Q_0 = \sum_{\mathbf{n},\q:|\mathbf{n}|+|\q|=r}\lambda_{\mathbf{n},\q}(\gamma)P_{\mathbf{n}}Q_{\q} + \mathfrak{R}_{\delta},
\end{equation}
where $\lambda_{\mathbf{n},\q}(\gamma)$ are rational fractions in $\gamma$ and are independent of other parameters (the expression of $\lambda_{\mathbf{n},\q}(\gamma)$ is different when $\chi$ takes the value $\frac{\gamma}{2}$ or $\frac{2}{\gamma}$). To avoid ambiguity of the definition, we proceed as if the family $(P_{\n}Q_{\q})_{\n,\q}$ is linearly independent and regroup the coefficients to obtain the above equation. In section \ref{section proof BPZ}, we will prove that: every coefficient $\lambda_{\mathbf{n},\q}(\gamma)$ equals $0$. Then by sending $\delta$ to $0$, we have $\mathcal{D}_r \lg V_{-\frac{(r-1)\gamma}{2}}(z) \prod_l V_{\alpha_l}(z_l)\rg = 0$ in the weak sense. The smoothness of correlation functions allows to conclude the proof for Theorem \ref{theo BSA}.

Now we discuss the first step of proving Proposition \ref{prop L_-nQ}. It is easy to see that 
\begin{equation}
L_{-n}P_{\n}Q_{\q} = (\sum_i n_i \frac{P_{n_i+n}}{P_{n_i}})P_{\n}Q_{\q} +P_{\n}L_{-n}Q_{\q}.
\end{equation}
Without loss of generality, we can consider $P_{\n}=1$. Let us first prove an intermediate lemma, which is the special case where $\q=0$:
\begin{lem}\label{lemma L_nQ}
The following relation holds when $n \le r$:
\begin{align}
	L_{-n}Q_0=&\Big(  -\sum_{i=1}^{n-1}P_iP_{n-i}+ ( (n-1)Q - (r-1)\chi)P_n\Big)Q_0 +2\sum_{i=1}^{n-1}P_iQ_{n-i}+((r-1)\chi-\frac{2(n-1)}{\gamma})Q_n  \nonumber\\
	&-\sum_{i=1}^{n-1} Q_{i,n-i}+ \mu \intc \frac{\partial_z\theta_{\delta}(y-z) }{(y-z)^{n-1}} \lg z,\z ;y\rg_{\delta} \,d^2y.
\end{align}
\end{lem}
\begin{rem}
For $\epsilon \ge 0$, the notation $\lg z,\z;\y \rg_{\delta,\epsilon}$ stands for $\lg V_{-(r-1)\chi,\epsilon}(z) \prod_{i=1}^p V_{\gamma,\epsilon}(y_i)\prod_{l=1}^N V_{\alpha_l,\epsilon}(z_l)\rg_{\delta}$. 
\end{rem}
\begin{proof}
We will need to use Lemma \ref{lem derivation G} for taking derivatives of $\lg V_{-(r-1)\chi,\epsilon}(z)  \prod_{l=1}^N V_{\alpha_l,\epsilon}(z_l)\rg_{\delta}$, and then tend $\epsilon$ to $0$ to get the desired relation. Remark that all the convergences in this proof are locally uniform convergences for $(z,\z)$ when $\epsilon \to 0$, and it will not be specified.

Let us work with the case $n \ge 2$,
\begin{small}
\begin{align}\label{equation AB}
	L_{-n}\lg z,\z\rg_{\delta,\epsilon} =&\left(-\sum_j\sum_{l \neq j}\frac{\alpha_j\alpha_l}{2(z_j-z)^{n-1}(z_l-z_j)_{\epsilon}}-\sum_j \frac{(r-1)\chi\alpha_j}{2(z_j-z)^{n-1}(z_j-z)_{\epsilon}}+\sum_j\frac{(n-1)\Delta_{\alpha_j}}{(z_j-z)^n}\right) \lg z,\z \rg_{\delta,\epsilon}  \nonumber\\
	& +\sum_j \frac{\mu \gamma\alpha_j}{2(z_j-z)^{n-1}}\intc \frac{\theta_{\delta}(y-z)}{(y-z_j)_{\epsilon}}\lg z,\z;y \rg_{\delta,\epsilon} \,d^2y \nonumber\\
	=&:A_{\epsilon} \lg z,\z \rg_{\delta,\epsilon}+B_{\epsilon}.
\end{align}
\end{small}
We use a simple identity to calculate $\lim_{\epsilon\to 0}A_{\epsilon}$ and $\lim_{\epsilon\to 0}B_{\epsilon}$:
	\begin{equation}\label{multisum}
	\frac{1}{(x_1-x_2)(x_2-z)^{n-1}}-\frac{1}{(x_1-x_2)(x_1-z)^{n-1}}=\sum_{i=1}^{n-1} \frac{1}{(x_1-z)^i}\frac{1}{(x_2-z)^{n-i}}.
	\end{equation}
By symmetry and the above identity,
	\begin{equation*}
	\begin{split}
	\sum_j\sum_{l \neq j}\frac{\alpha_j\alpha_l}{2(z_l-z_j)(z_j-z)^{n-1}}&=\frac{1}{2}\sum_j\sum_{l \neq j}\frac{\alpha_j\alpha_l}{2(z_l-z_j)(z_j-z)^{n-1}}-\frac{\alpha_j\alpha_l}{2(z_l-z_j)(z_l-z)^{n-1}}\\
	&=\sum_{i=1}^{n-1} P_iP_{n-i}-\sum_j \frac{(n-1)\alpha_j^2}{4(z_j-z)^n}.
	\end{split}
	\end{equation*}
Therefore taking the limit for $A_{\epsilon}$ yields
	\begin{equation*}
	\lim_{\epsilon\to 0}A_{\epsilon}=-\sum_{i=1}^{n-1} P_iP_{n-i}+\left( (n-1)Q-(r-1)\chi  \right)P_n.
	\end{equation*}
For $B_{\epsilon}$, note that
	\begin{equation*}
	\begin{split}
	&\quad B_{\epsilon}-\intc \sum_j\frac{\mu\gamma\alpha_j \theta_{\delta}(y-z) }{2(y-z_j)_{\epsilon}(y-z)^{n-1}}\lg z,\z;y \rg_{\delta,\epsilon} \,d^2y\\
	&= \sum_{i=1}^{n-1}\sum_j  \frac{\mu\gamma\alpha_j}{2(z_j-z)^i} \intc \frac{y-z_j}{(y-z_j)_{\epsilon}}\frac{\theta_{\delta}(y-z)}{(y-z)^{n-i}}\lg z,\z ;y\rg_{\delta,\epsilon}\,d^2y \overset{\epsilon \to 0}{\longrightarrow} 2\sum_{i=1}^{n-1}P_iQ_{n-i}.
	\end{split}
	\end{equation*}
Here we have used dominant convergence, where we can bound $\frac{x}{(x)_{\epsilon}}<c$ with $c$ a constant independent of $\epsilon$, and the function $y \mapsto \sup_{\epsilon}\lg z,\z ;y\rg_{\delta,\epsilon}$ is integrable. 

An integration by parts formula (or Stokes formula) gives the following identity:
	\begin{equation*}
	\begin{split}
	&\intc \sum_j\alpha_j \frac{\mu\gamma \theta_{\delta}(y-z) }{2(y-z_j)_{\epsilon}(y-z)^{n-1}}\lg z,\z ;y\rg_{\delta,\epsilon} \,d^2y\\
	= &-\frac{2(n-1)}{\gamma}\int_{\mathbb{C}} \frac{\mu \gamma \theta_{\delta}(y-z) }{2(y-z)^{n}} \lg z,\z ;y\rg_{\delta,\epsilon} \,d^2y+(r-1)\chi\int_{\mathbb{C}} \frac{\mu \gamma \theta_{\delta}(y-z) }{2(y-z)^{n-1}(y-z)_{\epsilon}} \lg z,\z ;y\rg_{\delta,\epsilon} \,d^2y  \\
	&+\mu \intc \frac{\partial_z\theta_{\delta}(y-z) }{(y-z)^{n-1}}\lg z,\z ;y\rg_{\delta,\epsilon} \,d^2y
	-2\int_{\mathbb{C}^2} \frac{(\mu \gamma)^2\theta_{\delta}(x-z)\theta_{\delta}(y-z)}{4(x-y)_{\epsilon}(y-z)^{n-1}}\lg z,\z ;x,y\rg_{\delta,\epsilon} \, d^2x d^2y.
	\end{split}
	\end{equation*}
The first two terms in the sum converge to $((r-1)\chi -\frac{2(n-1)}{\gamma}) Q_n$. For the last term, using the integrability of $(x,y) \mapsto \sup_{\epsilon}\lg z,\z ;x,y\rg_{\delta,\epsilon}$ and by symmetry,
\begin{align*}
&\quad 2\int_{\mathbb{C}^2} \frac{(\mu \gamma)^2\theta_{\delta}(x-z)\theta_{\delta}(y-z)}{4(x-y)_{\epsilon}(y-z)^{n-1}}\lg z,\z ;x,y\rg_{\delta,\epsilon} \, d^2x d^2y \\
&=\sum_{i=1}^{n-1}\int_{\mathbb{C}^2}\frac{(\mu \gamma)^2}{4} \frac{x-y}{(x-y)_{\epsilon}} \frac{\theta_{\delta}(x-z)\theta_{\delta}(y-z)}{(y-z)^i(x-z)^{n-i}}\lg z,\z ;x,y\rg_{\delta,\epsilon} d^2x d^2y\overset{\epsilon \to 0}{\longrightarrow} \sum_{i=1}^{n-1} Q_{i,n-i}.
\end{align*}
From the above calculus, we deduce that
	\begin{equation*}
	\begin{split}
	\lim_{\epsilon\to 0}B_{\epsilon}=&2\sum_{i=1}^{n-1}P_iQ_{n-i}+((r-1)\chi-\frac{2(n-1)}{\gamma})Q_n-\sum_{i=1}^{n-1}Q_{i,n-i} +\mu \intc \frac{\partial_z\theta_{\delta}(y-z) }{(y-z)^{n-1}} \lg z,\z ;y\rg_{\delta} \,d^2y.
	\end{split}
	\end{equation*}
Sending $\epsilon$ to $0$ in $L_{-n}\lg z,\z \rg_{\delta,\epsilon} = A_{\epsilon}\lg z,\z \rg_{\delta,\epsilon}+B_{\epsilon}$ proves the lemma for $n \ge 2$ in the weak derivative sense. Then it suffices to conclude with the smoothness of correlations. It is not hard to verify the validity for the case $n=1$, which concludes the proof.
\end{proof}

We have shown Proposition \ref{prop L_-nQ} in the special case $\q = 0$ in the previous lemma. The proof for the general case can then be deduced from this result.

\begin{proof}[Proof of Proposition \ref{prop L_-nQ}]
Consider the case $n\ge 2$. Let $\y = (y_1,\dots,y_p)$ and $R > 0$. Note that the operator $L_{-n}$ commutes with the integral sign in the following expression:
	\begin{equation*}
L_{-n}(\frac{\mu \gamma}{2})^p \int_{B(0,R)^p} \prod_{j=1}^p{\frac{\theta_{\delta}(y_j-z)}{(y_j-z)^{q_j}}} \lg z,\z;\y \rg_{\delta,\epsilon} \, d^2\mathbf{y}=(\frac{\mu \gamma}{2})^p \int_{B(0,R)^p} \prod_{j=1}^p{\frac{\theta_{\delta}(y_j-z)}{(y_j-z)^{q_j}}} L_{-n} \lg z,\z;\y \rg_{\delta,\epsilon} \,d^2\mathbf{y}.	
	\end{equation*}
We introduce the notation
\begin{equation*}
	L_{-n}^{\lg p \rg}=\sum_{l=1}^N \left(-\frac{\partial_{z_l}}{(z_l-z)^{n-1}}+\frac{(n-1)\Delta_{\alpha_l}}{(z_l-z)^n}\right)+\sum_{i=1}^p \left(-\frac{\partial_{y_i}}{(y_i-z)^{n-1}}+\frac{(n-1)\Delta_{\gamma}}{(y_i-z)^n}\right).
\end{equation*}
This newly defined operator considers $\y$ as insertions and applies the corresponding differential operators. Remark that the value of $\Delta_{\gamma}$ is $1$ and we can write
\begin{small}
\begin{align}
&(\frac{\mu \gamma}{2})^p \int_{B(0,R)^p} \prod_{j=1}^p{\frac{\theta_{\delta}(y_j-z)}{(y_j-z)^{q_j}}}L_{-n}\lg z,\z;\y \rg_{\delta,\epsilon} \,d^2\mathbf{y} \nonumber\\
=& (\frac{\mu \gamma}{2})^p\int_{B(0,R)^p} \prod_{j=1}^p{\frac{\theta_{\delta}(y_j-z)}{(y_j-z)^{q_j}}}L_{-n}^{\lg p \rg}\lg z,\z;\y \rg_{\delta,\epsilon} \, d^2\mathbf{y} \nonumber \\
& +(\frac{\mu \gamma}{2})^p\int_{B(0,R)^p} \prod_{j=1}^p{\frac{\theta_{\delta}(y_j-z)}{(y_j-z)^{q_j}}}\sum_{i=1}^p \big(\frac{\partial_{y_i}}{(y_i-z)^{n-1}}-\frac{n-1}{(y_i-z)^n}\big)\lg z,\z;\y \rg_{\delta,\epsilon} \, d^2\mathbf{y} \nonumber\\
 =:& \tilde{A}_{R,\epsilon}+\tilde{B}_{R,\epsilon}.
\end{align}
\end{small}
By the previous lemma, when $\epsilon \to 0$, $\tilde{A}_{R,\epsilon}$ converges to
$$ \z \mapsto (\frac{\mu \gamma}{2})^p \int_{B(0,R)^p} \prod_{j=1}^p{\frac{\theta_{\delta}(y_j-z)}{(y_j-z)^{q_j}}} L_{-n}^{\lg p \rg} \lg z,\z;\y \rg_{\delta} \,d^2\mathbf{y}$$
in the sense of distributions. This is because the integral can be regarded as integrating $L_{-n}^{\lg p \rg} \lg z,\z;\y \rg_{\delta,\epsilon}$ against a test function of $\y$. From the expression of $L_{-n}^{\lg p \rg} \lg z,\z;\y \rg_{\delta}$ (see Lemma \ref{lemma L_nQ}) we can see that it does not introduce any singularity for the integral. Consider for example the integral below:
\begin{align*}
(\frac{\mu \gamma}{2})^p \left| \int_{B(0,R)^c}\int_{B(0,R)^{p-1}} \prod_{j=1}^p{\frac{\theta_{\delta}(y_j-z)}{(y_j-z)^{q_j}}} L_{-n}^{\lg p \rg} \lg z,\z;\y \rg_{\delta} \,d^2\mathbf{y} \right|.
\end{align*} 
We can bound it simply by
\begin{align*}
\frac{1}{R^{q_1}}(\frac{\mu \gamma}{2})^p \int_{\mathbb{C}^{p}} \theta_{\delta}(y_1-z) \prod_{j=2}^p{\frac{\theta_{\delta}(y_j-z)}{|y_j-z|^{q_j}}} \left| L_{-n}^{\lg p \rg} \lg z,\z;\y \rg_{\delta} \right| \,d^2\mathbf{y}.
\end{align*}
The above term is well defined and by sending $R \to \infty$ it converges to $0$ for fixed $\delta$. Therefore we can write in the weak sense:
\begin{equation}
\tilde{A} := \lim_{R \to \infty} \tilde{A}_{R,0} = (\frac{\mu \gamma}{2})^p \int_{\mathbb{C}^p} \prod_{j=1}^p{\frac{\theta_{\delta}(y_j-z)}{(y_j-z)^{q_j}}} L_{-n}^{\lg p \rg} \lg z,\z;\y \rg_{\delta} \,d^2\mathbf{y}.
\end{equation}

On the other hand, an integration by parts shows that
\begin{equation*}
\tilde{B}_{R,\epsilon}= \sum_{i=1}^p (\frac{\mu \gamma}{2})^p\int_{B(0,R)^p} \prod_{j:j\neq i}{\frac{\theta_{\delta}(y_j-z)}{(y_j-z)^{q_j}}} \left(\frac{q_i \theta_{\delta}(y_i-z)}{(y_i-z)^{q_i+n}} -\frac{\partial_z\theta_{\delta}(y_i-z)}{(y_i-z)^{q_i+n-1}}\right)\lg z,\z;\y \rg_{\delta,\epsilon} \, d^2\mathbf{y} + O_{R \to \infty}(R^{-1}),
\end{equation*}
where the $O_{R \to \infty}(R^{-1})$ comes from the boundary term and can be bounded independently of $\epsilon$. Therefore by first sending $\epsilon$ to $0$ and then $R \to \infty$, we obtain the limit which equals:
\begin{equation}
\tilde{B} := \sum_{i=1}^p q_i T_i^n Q_{\mathbf{q}} + \mathfrak{R}_{\delta}.
\end{equation}
The above arguments show that 
\begin{align*}
L_{-n}(\frac{\mu \gamma}{2})^p \int_{\mathbb{C}^p} \prod_{j=1}^p{\frac{\theta_{\delta}(y_j-z)}{(y_j-z)^{q_j}}} \lg z,\z;\y \rg_{\delta} \, d^2\mathbf{y} = \tilde{A} + \tilde{B}.
\end{align*}

In the rest of this proof we do not need to take regularizations with $\eta_{\epsilon}$. With Lemma \ref{lemma L_nQ}, we calculate $\tilde{A}$:
\begin{small}
	\begin{equation*}
	\begin{split}
L_{-n}^{\lg p \rg}\lg z,\z;\y \rg_{\delta}
	=&\Big(  -\sum_{i=1}^{n-1}P_i^{\lg p \rg}P_{n-i}^{\lg p \rg}+ ( (n-1)Q -(r-1)\chi )P_n^{\lg p \rg}\Big)Q^{\lg p \rg}_0 +2\sum_{i=1}^{n-1}P_i^{\lg p \rg}Q_{n-i}^{\lg p \rg}\\
	&+((r-1)\chi-\frac{2(n-1)}{\gamma})Q_n^{\lg p \rg}-\sum_{i=1}^{n-1} Q_{i,n-i}^{\lg p \rg}
	+\mu \intc \frac{\partial_z\theta_{\delta}(y_{p+1}-z) }{(y_{p+1}-z)^{n-1}} \lg z,\z;\y,y_{p+1} \rg_{\delta} \,d^2y_{p+1},
	\end{split}
	\end{equation*}
\end{small}
where 
\begin{equation*}
P_k^{\lg p \rg}(z,\mathbf{z},\mathbf{y})=\sum_{l=1}^N\frac{\alpha_l}{2(z_l-z)^k}+\sum_{j=1}^p\frac{\gamma}{2(y_j-z)^k}=P_k+\sum_{j=1}^p\frac{\gamma}{2(y_j-z)^k},
\end{equation*}
\begin{equation*}
Q^{\lg p \rg}_0(z,\mathbf{z},\mathbf{y})=\lg V_{-(r-1)\chi}(z)\prod_{i=1}^p V_{\gamma}(y_i) \prod_{l=1}^N V_{\alpha_l}(z_l)\rg_{\delta}.
\end{equation*}
Hence
\begin{small}
\begin{align*}
& \tilde{A}= \Bigg[ ( (r-1)Q - (r-1)\chi )P_n+2\sum_{i=1}^{n-1}P_iT_{p+1}^{n-i} + ((r-1)\chi-\frac{2(n-1)}{\gamma})T_{p+1}^n-\sum_{i=1}^{n-1}T_{p+2}^{n-i} T_{p+1}^i\\
&+\sum_{j=1}^p \Bigg(-\gamma\sum_{i=1}^{n-1}P_iT_j^{n-i}+(\frac{(n-1)\gamma Q}{2} -\frac{(r-1)\gamma\chi}{2} )T_j^n\nonumber + \gamma\sum_{i=1}^{n-1}T_{p+1}^{n-i}T_j^i - \frac{\gamma^2}{4} \sum_{j'=1}^p \sum_{i=1}^{n-1}T_{j'}^{n-i}T_{j}^i \Bigg) \Bigg] P_{\n}Q_{\mathbf{q}} + \mathfrak{R}_{\delta}.
\end{align*}
\end{small}
This allows to prove the statement when $P_{\n} = 1$. And as remarked previously, this suffices to prove the statement for any $P_{\n}$.

Otherwise, one can verify the validity of the formula for the case $n=1$. This finishes the proof. 
\end{proof}

\subsection{Illustration with order 2 and 3}\label{section illustration}

We give the commutation relation for $L_{-n}\,(n\ge 1)$, which can be easily verified with Definition \ref{def L}:
\begin{lem}
For $n,m \ge 1$
\begin{equation}\label{commutation}
[L_{-n},L_{-m}]=(m-n)L_{-(n+m)}
\end{equation}
\end{lem}

Now we check BPZ equations of order $r=2$ and $r=3$ with $\gamma \in (0,2)$. We will see that the proof of the BPZ equations becomes rather simple and involves only algebraic simplifications.\\

\noindent$\Diamond$ $r=2$: By definition,
\begin{equation}
\mathcal{D}_{2}=\chi^2 L_{-2}+L_{-1}^2.
\end{equation}
With the help of Proposition \ref{prop L_-nQ}, we calculate:
	\begin{equation*}
	L_{-1}Q_0=-\chi P_1Q_0+\chi Q_1 + \mathfrak{R}_{\delta}.
	\end{equation*}
By applying the operator $L_{-1}$ to the above equation, we obtain
	\begin{equation*}
	L_{-1}^2Q_0
	=(-\chi P_2+\chi^2 P_1^2)Q_0-2\chi^2 P_1Q_1+\chi^2 Q_{1,1}+\chi(-\frac{\gamma\chi}{2}+1)Q_2+\mathfrak{R}_{\delta}.
	\end{equation*}
Again by Proposition \ref{prop L_-nQ}, we calculate:
	\begin{equation*}
	L_{-2}Q_0=(-P_1^2 +\frac{1}{\chi}P_2)Q_0+2P_1Q_1+(\chi-\frac{2}{\gamma})Q_2-Q_{1,1}+\mathfrak{R}_{\delta}.
	\end{equation*}
We can verify easily that in $\mathcal{D}_{2} Q_0 = \chi^2 L_{-2}Q_0+L_{-1}^2Q_0$, all the coefficients before $P_{\mathbf{n}}Q_{\mathbf{q}}$ cancel and therefore $\mathcal{D}_{2}Q_0=\mathfrak{R}_{\delta}$. This allows to show BPZ equations in the weak sense, we can then conclude with the smoothness of correlation functions that the equation holds in the strong sense. \qed\\

\noindent $\Diamond$ $r=3$: By definition,
	\begin{equation}
	\mathcal{D}_{3}=\chi^4 L_{-3}+ \frac{\chi^2}{2}L_{-1}L_{-2}+\frac{\chi^2}{2}L_{-2}L_{-1}+\frac{1}{4}L_{-1}^3.
	\end{equation}
We have by lemma \ref{commutation}:
	\begin{equation*}
	L_{-2}L_{-1}=L_{-1}L_{-2}-L_{-3}.
	\end{equation*}
Then we can write
	\begin{equation}\label{equation D3}
	\mathcal{D}_{3}= (\chi^4 -\frac{\chi^2}{2}) L_{-3} +L_{-1}(\chi^2 L_{-2} + \frac{1}{4}L_{-1}^2).
	\end{equation}
Using Proposition \ref{prop L_-nQ}:
	\begin{equation*}
	L_{-1}^2Q_0=(-2\chi P_2+4 \chi^2 P_1^2)Q_0-8\chi^2 P_1Q_1+4\chi^2 Q_{1,1}+2\chi(-\gamma\chi+1)Q_2+\mathfrak{R}_{\delta},
	\end{equation*}
	\begin{equation*}
	L_{-2}Q_0=(-P_1^2+(\frac{1}{\chi} - \chi)P_2)Q_0+2P_1Q_1+(2\chi-\frac{2}{\gamma})Q_2-Q_{1,1}+\mathfrak{R}_{\delta}.
	\end{equation*}
Hence
	\begin{equation*}
	(\chi^2L_{-2}+\frac{1}{4}L_{-1}^2)Q_0=(\frac{\chi}{2}-\chi^3)P_2Q_0+(\chi^3-\frac{\chi}{2})Q_2+\mathfrak{R}_{\delta}.
	\end{equation*}
We can then write
\begin{align*}
\mathcal{D}_{3} Q_0 = (\chi^4 -\frac{\chi^2}{2}) L_{-3} Q_0 + (\frac{\chi}{2}-\chi^3)L_{-1}P_2Q_0+(\chi^3-\frac{\chi}{2}) L_{-1}Q_2+\mathfrak{R}_{\delta}.
\end{align*}
We mention that when $\chi = \frac{\gamma}{2}$ and $\gamma \in (0,1] $, there is a type of $\mathfrak{R}_{\delta}$ that does not vanish when $\delta \to 0$:
\begin{equation}\label{equation R singular}
\int_{\mathbb{C}}  \frac{ \partial_z\theta_{\delta}(y-z)}{(y-z)^{2}} \lg V_{-(r-1)\chi }(z)  V_{\gamma}(y) \prod_{l=1}^N V_{\alpha_l}(z_l)\rg_{\delta} \,d^2 y.
\end{equation}
For this kind of non-vanishing perturbation term, we will write it directly with its expression instead of writing $\mathfrak{R}_{\delta}$. In general, there is another type 
$$\int_{\mathbb{C}^2}  \frac{ \partial_z\theta_{\delta}(y_1-z)}{y_1-z} \frac{\theta_{\delta}(y_2-z)}{y_2-z} \lg V_{-(r-1)\chi }(z)  V_{\gamma}(y_1) V_{\gamma}(y_2) \prod_{l=1}^N V_{\alpha_l}(z_l)\rg_{\delta} \,d^2 y_1 d^2 y_2$$
that does not vanish. But thanks to the specific replication we use  for $\mathcal{D}_3$ \eqref{equation D3}, this term will not appear in the final expression, .

For the other types of perturbation terms, they still converge to $0$, and we will keep the notation $\mathfrak{R}_{\delta}$. With our calculus, in the expression of $\mathcal{D}_3 Q_0$ the term \eqref{equation R singular} appears only in $(\chi^3-\frac{\chi}{2})L_{-1}Q_2$ and in $(\chi^4 -\frac{\chi^2}{2}) L_{-3}Q_0$. We can find the exact form of the perturbation term in the proof of Proposition \ref{prop L_-nQ}. More precisely, we have
	\begin{align*}
	&(\frac{\chi}{2}-\chi^3)L_{-1}P_2Q_0+(\chi^3-\frac{\chi}{2}) L_{-1}Q_2\\
	=&-(\chi^4 -\frac{\chi^2}{2})\Big((\frac{2}{\chi}P_3-2P_1P_2)Q_0+2P_2Q_1+2P_1Q_2-2Q_{2,1}+(\gamma-\frac{2}{\chi})Q_3\Big) \\
	&- \mu \frac{\gamma}{2}(\chi^3-\frac{\chi}{2}) \int_{\mathbb{C}}  \frac{ \partial_z\theta_{\delta}(y-z)}{(y_1-z)^{2}} \lg V_{-(r-1)\chi }(z)  V_{\gamma}(y) \prod_{l=1}^N V_{\alpha_l}(z_l)\rg_{\delta} \,d^2\mathbf{y} + \mathfrak{R}_{\delta},
	\end{align*}
and
\begin{align*}
L_{-3}Q_0=&(\frac{2}{\chi}P_3-2P_1P_2)Q_0+2P_2Q_1+2P_1Q_2-2Q_{2,1}+(2\chi-\frac{2}{\gamma})Q_3 \\
	&+ \mu  \int_{\mathbb{C}}  \frac{ \partial_z\theta_{\delta}(y-z)}{(y_1-z)^{2}} \lg V_{-(r-1)\chi }(z)  V_{\gamma}(y) \prod_{l=1}^N V_{\alpha_l}(z_l)\rg_{\delta} \,d^2\mathbf{y} + \mathfrak{R}_{\delta}.
\end{align*}
When $\chi = \frac{\gamma}{2}$, we can verify that all the terms cancel and $\mathcal{D}_{3}Q_0=\mathfrak{R}_{\delta}$. Especially, the perturbations that we cannot control cancel among them. When $\chi = \frac{2}{\gamma}$, the term $$\int_{\mathbb{C}}  \frac{ \partial_z\theta_{\delta}(y-z)}{(y_1-z)^{2}} \lg V_{-(r-1)\chi }(z)  V_{\gamma}(y) \prod_{l=1}^N V_{\alpha_l}(z_l)\rg_{\delta} \,d^2\mathbf{y}$$ converges to $0$ by Proposition \ref{prop R to 0} and we can keep using the notation $\mathfrak{R}_{\delta}$ for it. Hence we also have $\mathcal{D}_{3}Q_0=\mathfrak{R}_{\delta}$. This finishes the proof for BPZ equations of order 3.  \qed

\subsection{Proof of the BPZ equations of order r with real Coulomb gas}\label{section proof BPZ}

According to the previous discussions, the proof of the BPZ equations of order $r$ has been reduced to an algebraic problem. Interestingly, real coulomb gas integrals with a degenerate insertion satisfy the same recursive relations, but without perturbation terms. 

\begin{definition}
For $x_1 < \dots < x_l \,(l \ge 1)$ and $t < t_1 < \dots < t_N\, (N\ge 2)$, we denote the integrand of real Coulomb gas integrals with degenerate insertions as
\begin{equation}
f^{(l)}_{-(r-1)\chi,\alp}(t,\mathbf{t};\x):= \prod_{0\le i< j \le N} (t_j-t_i)^{-\frac{\alpha_i\alpha_j}{2}} \prod_{0 \le i \le N , 1\le s \le l}(x_s-t_i)^{-\frac{\gamma\alpha_i}{2}}\prod_{1\le s'<s \le l }(x_s-x_{s'})^{-\frac{\gamma^2}{2}},
\end{equation}
where we denote $\alpha_0 = -(r-1)\chi, t_0=t$ and by convention $(-1)^{\alpha} = e^{i\pi \alpha}$.
\end{definition}

For real Coulomb gas integrals, we will always work with the condition
\begin{equation}\label{equation real condition}
(\gamma,\alpha_{N-1},\alpha_N) \in (i \mathbb{R}_+)^{3},\min \{ -\frac{\gamma\alpha_{N-1}}{2},-\frac{\gamma\alpha_N}{2}   \} \ge r.
\end{equation}
It is easy to see that under this condition, $$C^{(l)}_{-(r-1)\chi,\alp}(t,\mathbf{t})=\int_{t_{N-1} < x_1< \dots < x_l < t_N} f^{(l)}_{-(r-1)\chi,\alp}(t,\mathbf{t};\x) d\x$$ is well defined and at least $\mathcal{C}^r$.

Next we define the equivalent of $Q_{\q}$ for Coulomb gas integrals.

\begin{definition}\label{definition Coulomb Q}
Let $p\in \mathbb{N}$, and $\mathbf{q}=(q_1, \dots, q_p) \in (\mathbb{N}^*)^p$, we define
\begin{equation}
Q_{\mathbf{q}}^{(l)}(t,\mathbf{t}):= \int_{t_{N-1} < x_1 < \dots < x_l < t_N} R_{\q}^{(l)}(t;\x) f^{(l)}_{-(r-1)\chi,\alp}(t,\mathbf{t};\x) d^2 \x.
\end{equation} 
where
\begin{equation}
R_{\q}^{(l)}(t;\x) := (-\frac{\gamma}{2})^p \sum_{1\le s_1< \dots < s_p \le l}\prod_{j=1}^p \frac{1}{(x_{s_j}-t)^{q_j}}.
\end{equation}
The operator $T_k$ on $Q_{\mathbf{q}}^{(l)}$ with $k \in \mathbb{N}^*$ is defined as:
\begin{equation}
T_k Q_{\mathbf{q}}^{(l)}=Q^{(l)}_{q_1, \dots, q_k+1, \dots, q_p}.
\end{equation}
\end{definition}

\begin{rem}
By convention $Q_{\q}^{(l)}=0$ if $p>l$. Note that in the expression of $Q_{\q}^{(l)}$ there is no need for regularization $\theta_{\delta}$ around $x_i=t$, since we are considering $x_i > t_{N-1} >t$ so that $t$ is not a singularity.
\end{rem}

By abuse of notation, when dealing with real variables, $L_{-n}$ is defined as a real differential operator:
\begin{equation}\label{equation real L}
	L_{-1}=\partial_{t},\quad L_{-n}=\sum_{l=1}^N \left(-\frac{1}{(t_l-t)^{n-1}}\partial_{t_l}+\frac{\Delta_{\alpha_l}(n-1)}{(t_l-t)^n}\right) \quad n\ge 2.
\end{equation}
Then the same recursive relation holds for real Coulomb gas integrals:

\begin{prop}
$(P_{\n}(t,\mathbf{t})Q_{\q}^{(l)}(t,\mathbf{t}))_{\n,\q}$  satisfies Proposition \ref{prop L_-nQ} with no perturbation terms $\mathfrak{R}_{\delta}$. 
\end{prop}
\begin{rem}
There is no $\mathfrak{R}_{\delta}$ since there is no regularization $\theta_{\delta}$. The proof follows exactly the same steps as the proof of Proposition \ref{prop L_-nQ} and there is no need for the regularization $\eta_{\epsilon}$ to help calculate the derivatives.
\end{rem}

The proposition tells in particular that 
\begin{equation}
\mathcal{D}_r Q^{(l)}_0 = \sum_{\mathbf{n},\q:|\mathbf{n}|+|\q|=r}\lambda_{\mathbf{n},\q}(\gamma)P_{\mathbf{n}}Q^{(l)}_{\q},
\end{equation}
with the same coefficients $\lambda_{\n,\q}(\gamma)$ as introduced in \eqref{equation lambda}.

To prove BPZ equations, we only need to show that all the $\lambda_{\n,\q}(\gamma)$ are equal to $0$. We know that they are rational fractions of $\gamma$, so it suffices to prove it for an infinity number of values for $\gamma \in i\mathbb{R}_+$. We show in the following that real Coulomb gas integrals actually satisfy BPZ equations, which allows to solve the combinatorial problem, see Proposition \ref{prop lambda=0}. 

\begin{lem}\label{lem BSA f}
For $r \in \mathbb{N}^*$, the following differential equation holds:
\begin{equation}
\mathcal{D}_{r} f^{(0)}_{-(r-1)\chi,\alp}(t,\mathbf{t})=0.
\end{equation}
\end{lem}

\begin{rem}
Here $D_r$ is composed of real differential operators $L_{-n}$, see \eqref{equation real L}. The proof of this lemma can be found in section \ref{section integrand}. Note that this result has been proved by Kytola-Peltola  \cite{Kytola2014conformally} with a fusion technique in \cite{Dubedat2015sle}. The fusion technique requires some non trivial manipulations of the Virasoro algebra. We would like to mention that our proof is purely combinatorial and elementary.
\end{rem}

\begin{prop}\label{prop BSA Coulomb}
Take $(\gamma,\alp)$ such that \eqref{equation real condition} is satisfied. Then the real Coulomb gas integrals verify BPZ equations of order $r$. More precisely, 
\begin{equation}
\mathcal{D}_r C^{(l)}_{-(r-1)\chi,\alp}(t,\mathbf{t}) = 0.
\end{equation}
\end{prop}
\begin{proof}
By applying derivation under the integral sign and Stokes Theorem,
\begin{equation*}
\begin{split}
\mathcal{D}_r C^{(l)}_{-(r-1)\chi,\alp}(t,\mathbf{t}) &=\int_{t_{N-1} < x_1< \dots < x_l < t_N} \mathcal{D}_r f^{(l)}_{-(r-1)\chi,\alp}(t,\mathbf{t};\x) d \x\\
&= \int_{t_{N-1} <x_1< \dots < x_l < t_N} \mathcal{D}_r^{\lg l \rg}f^{(l)}_{-(r-1)\chi,\alp}(t,\mathbf{t};\x)d \x,
\end{split}
\end{equation*}
with $\mathcal{D}_r^{\lg l \rg}$ the operator $\mathcal{D}_r$ where we replace in its expression the operators $ L_{-n}$ by $$L_{-n}^{\lg l \rg}:=L_{-n}+\sum_{s=1}^l  \left(-\frac{\partial_{x_s}}{(x_s-t)^{n}}+\frac{n-1}{(x_s-t)^n} \right) \quad n \ge 2.$$ By Lemma \ref{lem BSA f}, $\mathcal{D}_r^{\lg l \rg}f^{(l)}_{-(r-1)\chi,\alp}(t,\mathbf{t};\x)=0$, this shows that $\mathcal{D}_r C^{(l)}_{-(r-1)\chi,\alp}(t,\mathbf{t})=0$.

\end{proof}

\begin{prop}\label{prop lambda=0}
For $\n,\q$ such that $|\n|+|\q| =r$, the rational function $\lambda_{\n,\q}(\gamma)$ equals $0$.
\end{prop}
\begin{proof}
We will work under the condition \eqref{equation real condition} with $N$ sufficiently large. From the above proposition together with the discussion from the previous subsection, we deduce that $$\sum_{\n,\q:|\n|+|\q|=r} \lambda_{\n,\q}(\gamma)P_{\mathbf{n}}Q^{(r)}_{\q}=\mathcal{D}_r C^{(r)}_{-(r-1)\chi,\alp}(t,\mathbf{t}) = 0.$$

For simplicity, let us denote $t_0=t$ and $\alpha_0=-(r-1)\chi$. We can divide the left hand side of the equation by the common term $\prod_{0\le i< j \le N}(t_j-t_i)^{-\frac{\alpha_i\alpha_j}{2}}$, then we have
\begin{small}
\begin{align*}
\sum_{\n:|\n|\le r}P_{\mathbf{n}}\sum_{\q:|\q|=r-|\n|}\lambda_{\mathbf{n},\q}(\gamma)\int_{t_{N-1} < x_1< \dots < x_r < t_N} R_{\q}^{(r)}(t;\x)\prod_{0 \le i \le N , 1\le s \le r}(x_s-t_i)^{-\frac{\gamma\alpha_i}{2}}\prod_{1\le s<s' \le r }(x_s-x_{s'})^{-\frac{\gamma^2}{2}} d\x  = 0.
\end{align*}
\end{small}
Denote the left hand side by $(\star)$. We claim that the function $g_{\n}((t_i)_{r+1 \le  i \le N})$ defined by
\begin{align}
g_{\n}&((t_i)_{r+1 \le  i \le N}) :=\sum_{\q:|\q|=r-|\n|}\lambda_{\mathbf{n},\q}(\gamma)\int_{t_{N-1} < x_1 < \dots < x_r < t_N} R_{\q}^{(r)}(t,\x)  \nonumber\\
&\prod_{r+1 \le i \le N , 1\le s \le r}(x_s-t_i)^{-\frac{\gamma\alpha_i}{2}} \prod_{1\le s \le r}(x_s-t)^{-\frac{\gamma \sum_{i=0}^{r}\alpha_i}{2}}\prod_{1\le s<s' \le r }(x_s-x_{s'})^{-\frac{\gamma^2}{2}} d\x
\end{align}
equals $0$ for all $|\n| \le r$. To see this, suppose that for all $1 \le l \le r$, $\alpha_l \neq 0$. We study the asymptotic when $t_{1}, t_{2}, \dots , t_r$ tend to $t$ simultaneously:
\begin{align*}
(\star) &= \left(\sum_{\n:  |\n| \le r} \prod_{j=1}^m  \left(\sum_{l={1}}^r \frac{\alpha_l}{2 (t_l-t)^{n_j}} \right) g_{\n}((t_i)_{r+1 \le  i \le N})\right) (1+o(1))\\
&= \left(\sum_{\n: |\n|\le r}  \left(\sum_{\n': \n \subseteq \n'} c_{\n'}((\alpha_l)_{1 \le l \le r}) g_{\n'}((t_i)_{r+1 \le  i \le N})\right) \sum_{1 \le i_1 < \dots < i_m \le r}\frac{m!}{\prod_{j=1}^m (t_{i_j}-t)^{n_j}} \right) (1+o(1))   ,
\end{align*}
where $\n \subseteq \n'$ means that $\n$ is a sub-tuple of $\n'$ and in particular, when $\n' =\n$, $$c_{\n'}((\alpha_l)_{1 \le l \le r}) = \frac{\prod_{l=1}^r \alpha_l}{2^r}.$$
Since $(\star)$ also equals $0$, it is not hard to show (it should be done in a certain order) that  $$\forall |\n|\le r,\, \sum_{\n': \n \subseteq \n'} c_{\n'}((\alpha_l)_{1 \le l \le r}) g_{\n'}((t_i)_{r+1 \le i \le N})=0.$$
The equations above form a linear system with a triangular coefficient matrix with non null values on the diagonal, hence we will be able to conclude that $g_{\n}((t_i)_{r+1 \le i \le N}) = 0$ for all $|\n|\le r$. 

Now we have
\begin{align}
\sum_{\q:|\q|=r-|\n|}\lambda_{\mathbf{n},\q}(\gamma)\int_{t_{N-1} < x_1< \dots < x_r < t_N}& R_{\q}^{(r)}(t,\x)\prod_{r+1 \le i \le N , 1\le s \le r}(x_s-t_i)^{-\frac{\gamma\alpha_i}{2}}\nonumber\\
&\prod_{1\le s \le r}(x_s-t)^{-\frac{\gamma \sum_{i=0}^{r}\alpha_i}{2}} \prod_{1\le s<s' \le r }(x_s-x_{s'})^{-\frac{\gamma^2}{2}} d\x =0.
\end{align}
Let us take $\alpha_{r+1}=\dots = \alpha_{N-2}=-\frac{2}{\gamma}$, note that we can sum over a finite set $E_n$ of values of $t_i$ for each $r+1 \le i \le N-2$ to obtain
\begin{equation}
\sum_{t_i \in E_n}  \prod_{1\le s \le r}(x_s-t_i) =\sum_{1 \le i_1< \dots < i_n \le r} x_{i_1}\dots x_{i_n}= e_n(x_1,\dots, x_r). 
\end{equation}
Hence with a sum over $(t_{r+1},\dots, t_{N-2}) \in E_{n_{r+1}} \times \dots \times E_{n_{N-2}}$, we can obtain a product of fundamental symmetric polynomials:  $\prod_{j=r+1}^{N-2} e_{n_j}(x_1,\dots,x_r)$. Since $\sum_{\q:|\q|=r-|\n|}\lambda_{\mathbf{n},\q}(\gamma) R_{\q}^{(r)}(t,\x)\prod_{s=1}^r (x_s-t)^{r}$ is a symmetric polynomial in $(x_s)_{1\le s \le r}$, by the fundamental theorem of symmetric polynomials, when $N$ is sufficiently large we can sum up different values of $t_i \, ( r+1\le i \le N-2)$ to get
\begin{align}
\int_{t_{N-1} < x_1 < \dots < x_r < t_N} \left|\sum_{\q:|\q|=r-|\n|}\lambda_{\mathbf{n},\q}(\gamma)R_{\q}^{(r)}(t,\x)\right|^2\prod_{s=1}^r (x_s-t)^{r-\frac{\gamma \sum_{i=0}^{r}\alpha_i}{2}} \nonumber \\
\prod_{N-1 \le i \le N , 1\le s \le r}(x_s-t_i)^{-\frac{\gamma\alpha_i}{2}} \prod_{1\le s<s' \le r }(x_s-x_{s'})^{-\frac{\gamma^2}{2}} d\x =0.
\end{align}
This implies that for all $|\n| \le r$,
\begin{equation}
\sum_{\q:|\q|=r-|\n|}\lambda_{\mathbf{n},\q}(\gamma)R_{\q}^{(r)}(t,\x)=0.
\end{equation}
We can easily extend the above equations to all $x_s$ different from $t$. Then by a study of asymptotic when $x_s$ tend simultaneously to $t$, we conclude that $\lambda_{\n,\q}(\gamma)=0$.
\end{proof}

This finishes the proof of Theorem \ref{theo BSA}.

\subsection{BPZ equations for boundary Liouville theory}\label{section boundary BPZ}

Let us illustrate the idea with the degenerate insertion $B^{+}_{-(r-1)\chi}(t)$. The correlation function \eqref{equation boundary degenerate}  is holomorphic in $t$ in the upper half plane except the points $z_i$. The smoothness in $(\mathbf{t}, \z)$ in this case has not been proved, but the method in  \cite{Oikarinen2018Smoothness} applies to this case and we will assume this property. Note that the derivative in $L_{-1}^{\mathbb{H}}$ should be understood as a complex derivative with respect to $t$. 

If we think heuristically $V_{\alpha_i}(z_i) = B_{\alpha_i}(z_i) B_{\alpha_i}(\overline{z}_i)$, we can observe that the things behave very similar to the sphere case: firstly the form of $L_{-n}^{\mathbb{H}}$ is nothing but $L_{-n}$ written for the points $z_i,\overline{z}_i,t_j$, secondly we can observe the same form of derivative rule that we illustrate below. We use the regularization $\eta_{\epsilon}$ for the Gaussian free field $X$: for $x\in \mathbb{R}$
\begin{align*}
X_{\epsilon}(x) = 2\int_{y \in \mathbb{H}} X(y) \eta_{\epsilon}(x-y) d^2 y,
\end{align*}
and for $z_i, 1\le i \le N$, we define for $\epsilon$ sufficiently small
\begin{align*}
X_{\epsilon}(z_i) = X*\eta_{\epsilon}.
\end{align*}
By abuse of notation, for $1 \le i \neq j \le N$ and $x,x'\in \mathbb{R}$, we denote
\begin{align*}
\frac{1}{(z_i - z_j)_{\epsilon}} &= \intc \intc \frac{1}{z_i-z_j - x_1 + x_2} \eta_{\epsilon}(x_1) \eta_{\epsilon}(x_2) d^2 x_1 d^2 x_2,\\
\frac{1}{(x - x')_{\epsilon}} &= 4 \int_{\mathbb{H}} \int_{\mathbb{H}} \frac{1}{x-x' - x_1 + x_2} \eta_{\epsilon}(x_1) \eta_{\epsilon}(x_2) d^2 x_1 d^2 x_2,\\
\frac{1}{(x - z_i)_{\epsilon}} &= 2\int_{\mathbb{H}} \intc \frac{1}{x-z_i -x_1 + x_2} \eta_{\epsilon}(x_1) \eta_{\epsilon}(x_2) d^2 x_1 d^2 x_2.
\end{align*}
For the derivative rules, we have
\begin{align*}
&\partial_{z_k} \lg \prod_{i=1}^N V_{\alpha_i,\epsilon}(z_i)\prod_{j=1}^M B_{\beta_j,\epsilon}^{\mu_{j-1},\mu_j}(t_j) \rg_{\mathbb{H}} \\
=& \left(\sum_{i:i\neq k} \frac{\alpha_i\alpha_k}{2(z_i-z_k)_{\epsilon}} + \sum_{i}\frac{\alpha_i\alpha_k}{2(\overline{z}_i-z_k)_{\epsilon}}+\sum_j \frac{\beta_j\alpha_k}{2(t_j-z_k)_{\epsilon}} \right)\lg \prod_{i=1}^N V_{\alpha_i,\epsilon}(z_i)\prod_{j=1}^M B_{\beta_j,\epsilon}^{\mu_{j-1},\mu_j}(t_j) \rg_{\mathbb{H}}\\
&-\mu \int_{\mathbb{R}} \frac{\gamma\alpha_k}{2(y-z_k)_{\epsilon}} \lg B_{\gamma,\epsilon}(y)\prod_{i=1}^N V_{\alpha_i,\epsilon}(z_i)\prod_{j=1}^M B_{\beta_j,\epsilon}^{\mu_{j-1},\mu_j}(t_j) \rg_{\mathbb{H}} dy,
\end{align*}
and 
\begin{align*}
&\partial_{\overline{z}_k} \lg \prod_{i=1}^N V_{\alpha_i,\epsilon}(z_i)\prod_{j=1}^M B_{\beta_j,\epsilon}^{\mu_{j-1},\mu_j}(t_j) \rg_{\mathbb{H}} \\
=& \left(\sum_{i} \frac{\alpha_i\alpha_k}{2(z_i-\overline{z}_k)_{\epsilon}} + \sum_{i:i\neq k}\frac{\alpha_i\alpha_k}{2(\overline{z}_i-\overline{z}_k)_{\epsilon}}+\sum_j \frac{\beta_j\alpha_k}{2(t_j-\overline{z}_k)_{\epsilon}} \right) \lg \prod_{i=1}^N V_{\alpha_i,\epsilon}(z_i)\prod_{j=1}^M B_{\beta_j,\epsilon}^{\mu_{j-1},\mu_j}(t_j) \rg_{\mathbb{H}} \\
&-\mu \int_{\mathbb{R}} \frac{\gamma\alpha_k}{2(y-\overline{z}_k)_{\epsilon}} \lg B_{\gamma,\epsilon}(y)\prod_{i=1}^N V_{\alpha_i,\epsilon}(z_i)\prod_{j=1}^M B_{\beta_j,\epsilon}^{\mu_{j-1},\mu_j}(t_j) \rg_{\mathbb{H}} dy.
\end{align*}
Here the notation $B_{\gamma}(y)$ simply means that inserting $y$ between any $t_j$ and $t_{j+1}$ will keep the same boundary constant $s_j$ on both sides of $y$.  

Similarly, when deriving with respect to $t_j$:
\begin{align*}
\partial_{t_k} \lg \prod_{i=1}^N V_{\alpha_i,\epsilon}(z_i)\prod_{j=1}^M B_{\beta_j,\epsilon}^{\mu_{j-1},\mu_j}(t_j) \rg_{\mathbb{H}} = \sum_{i} \frac{\alpha_i\beta_k}{2(z_i-t_k)_{\epsilon}} + \sum_i\frac{\alpha_i\beta_k}{2(\overline{z}_i-t_k)_{\epsilon}}+\sum_{j:j\neq k} \frac{\beta_j\beta_k}{2(t_j-t_k)_{\epsilon}}\\
-\mu \int_{\mathbb{R}} \frac{\gamma\beta_k}{2(y-t_k)_{\epsilon}} \lg B_{\gamma,\epsilon}(y)\prod_{i=1}^N V_{\alpha_i,\epsilon}(z_i)\prod_{j=1}^M B_{\beta_j,\epsilon}^{\mu_{j-1},\mu_j}(t_j) \rg_{\mathbb{H}} dy.
\end{align*}
We observe that the derivative rules behave exactly as if we have insertions $z_i, \overline{z}_i, t_j$ on the sphere.

Now we present the analogies of $P_{\n}Q_{\q}$:
\begin{definition}
For $n\in \mathbb{N}^*$, we define
\begin{equation}
\mathcal{P}_n(t, \mathbf{t};\mathbf{z}):=\sum_{i=1}^N\left(\frac{\alpha_i}{2(z_i-t)^n}+\frac{\alpha_i}{2(\overline{z}_i-t)^n} \right) + \sum_{j=1}^M \frac{\beta_j}{2(t_j-t)^n}.
\end{equation}
Let $p\in \mathbb{N}$, and $\mathbf{q}=(q_1, \dots, q_p) \in (\mathbb{N}^*)^p$, we define
\begin{equation}
\mathcal{Q}_{\mathbf{q}}(z,\mathbf{z}):=(\frac{\mu \gamma}{2})^p \int_{\mathbb{R}^p} \prod_{l=1}^p\frac{1}{(y_l-t)^{q_l}} \lg B^{+}_{-(r-1)\chi}(t) \prod_{l=1}^p B_{\gamma}(y_l)\prod_{i=1}^NV_{\alpha_i}(z_i)\prod_{j=1}^M B_{\beta_j}^{\mu_{j-1},\mu_j}(t_j) \rg_{\mathbb{H}}\,d \mathbf{y}
\end{equation} 
We keep using the notation $T_k$ for $T_k{Q_{\mathbf{q}}}=Q_{q_1, \dots, q_k+1, \dots, q_p}$.
\end{definition}
\begin{rem}
There is no need for the regularization $\theta_{\delta}$ around $t$ because $y_l=t$ is not a singularity when $t$ is in the upper half plane. Therefore $\mathcal{Q}_{\q}$ are well defined objects.
\end{rem}

From the observations above, one can easily notice the following result:
\begin{prop}
$L_{-n}^{\mathbb{H}}$ with $\mathcal{P_{\n}} \mathcal{Q}_{\q}$ satisfy Proposition \ref{prop L_-nQ} without perturbation terms $\mathfrak{R}_{\delta}$.
\end{prop}

As a consequence, we can write 
\begin{align*}
D^{\mathbb{H}}_r \lg B^{+}_{-(r-1)\chi}(t) \prod_{i=1}^NV_{\alpha_i}(z_i)\prod_{j=1}^M B_{\beta_j}^{\mu_{j-1},\mu_j}(t_j) \rg_{\mathbb{H}} = \sum_{|\n|+|\q|=r}\lambda_{\n,\q}(\gamma)\mathcal{P}_{\n}\mathcal{Q}_{\q} = 0.
\end{align*}
Since there is no perturbation term, there is no constraint on $\gamma$. This finishes the proof of Theorem \ref{theo boundary BPZ} with the degenerate insertion $B^{+}_{-(r-1)\chi}(t)$. The case with $B^{-}_{-(r-1)\chi}(t)$ is exactly the same, and the case of BPZ equation with the degenerate insertion $V_{-(r-1)\chi}(z)$ can also be proved in the same manner. This allows to conclude the proof for Theorem \ref{theo boundary BPZ}.

\section{BPZ equations for the integrand}\label{section integrand}

Here we provide an elementary proof to show that the integrand of real Coulomb gas integrals satisfy BPZ equations. This result is the key element to show that real Coulomb gas integrals satisfy BPZ equations.

\begin{prop}
For $r \in \mathbb{N}^*$, the following differential equation holds:
\begin{equation}
\mathcal{D}_{r} f(t,\mathbf{t})=0,
\end{equation}
where
\begin{equation*}
f(t,\mathbf{t}):=f^{(0)}_{-(r-1)\chi,\alp}(t,\mathbf{t})= \prod_{1\le i< j \le N} (t_j-t_i)^{-\frac{\alpha_i\alpha_j}{2}} \prod_{i=1}^N (t_i-t)^{\frac{(r-1)\chi \alpha_i}{2}}
\end{equation*}
and recall that $L_{-n}$ for real variables are defined by
	\begin{equation}
	L_{-1}=\partial_{t}, \quad L_{-n}=\sum_{l=1}^N \left(-\frac{\partial_{t_l}}{(t_l-t)^{n-1}}+\frac{\Delta_l(n-1)}{(t_l-t)^n}\right) \quad n\ge 2.
	\end{equation}
$\Delta_l:=\frac{\alpha_l}{2}(Q-\frac{\alpha_l}{2})$ is the conformal weight. 
\end{prop}
\begin{proof}
First let us transform it into a combinatorial problem as what we have done in section \ref{section recursive formulas}. Let $P_{n}(t,\mathbf{t}) = \sum_{l=1}^N \frac{\alpha_l}{2(t_l-t)^n}$, we can show easily
\begin{equation}\label{equation Lf}
L_{-n}P_{\n}f=  \left(\sum_{i} n_i \frac{P_{n_i+n}}{P_{n_i}}-\sum_{i=1}^{n-1}P_iP_{n-i}+ ( (n-1)Q - (r-1)\chi )P_n\right) P_{\n}f
\end{equation}
As a consequence, $\mathcal{D}_r f = \left(\sum_{|\n|=r}\lambda_{\n}(\gamma)P_{\n}\right)f$, where $\lambda_{\n}(\gamma)$ corresponds to $\lambda_{\n,\q}(\gamma)$ with $\q=0$ (see \eqref{equation lambda}). Thus we will need to show that all the coefficients $\lambda_{\n}(\gamma)$ are zero.

Without loss of generality, we can restrict ourselves to the case $N\ge r$. Interestingly, we can largely simplify the problem if we take another point of view by treating $P_{\n}(t,\mathbf{t})$ as a polynomial of $(t,\mathbf{t})$ with values in the algebra $\mathbb{C}[\alp]$. If we quotient by the relation
\begin{equation}\label{equation quotient}
\forall 1\le i \le N,\, \alpha_{i}^2 = 0,
\end{equation}
then we can write
 $$P_{\n}=\prod_{i=1}^m P_{n_i} = \sum_{1 \le i_1 < \dots  i_m \le N} \frac{N!}{2^m(N-m)!}\frac{\alpha_{i_1} \dots \alpha_{i_m}}{(t_{i_1}-t)^{n_{i_1}}\dots (t_{i_m}-t)^{n_{i_m}}},$$
and
$$\mathcal{D}_r f = \left(\sum_{m=1}^r \sum_{\substack{(n_1,\,\dots,\,n_k)\in \mathbb{N^*}^k \\
		n_1+\dots+n_k=r}}\lambda_{\n}(\gamma)\sum_{1 \le i_1 < \dots  i_m \le N} \frac{N!}{2^m(N-m)!}\frac{\alpha_{i_1} \dots \alpha_{i_m}}{(t_{i_1}-t)^{n_{i_1}}\dots (t_{i_m}-t)^{n_{i_m}}}\right)f.$$
Thus if we can show $\mathcal{D}_r f =0$ under the quotient relation \eqref{equation quotient}, then by linear independence of functions $$(\alp,t,\mathbf{t})\mapsto \frac{\alpha_{i_1} \dots \alpha_{i_m}}{(t_{i_1}-t)^{n_{i_1}}\dots (t_{i_m}-t)^{n_{i_m}}},$$
we have that $\lambda_{\n}(\gamma)=0$ for all $|\n|=r$.

Now we prove $\mathcal{D}_r f =0$ under the condition \eqref{equation quotient}. In this setting, the operators $L_{-n}$ can be rewritten as 
\begin{equation}
L_{-n}:=\sum_{l=1}^N \left(-\frac{\partial_{t_l}}{(t_l-t)^{n-1}}+\frac{Q(n-1)\alpha_l}{2(t_l-t)^n}\right).
\end{equation} 
By first developing $L_{-n_1}f$ with the formula \eqref{equation Lf}, we have
\begin{align}\label{equation 3.6}
\mathcal{D}_r f =& \sum_{l_1=1}^N \sum_{k=1}^r \sum_{n_1+\dots+n_k=r} \frac{(\chi^2)^{r-k}}{\prod_{j=1}^{k-1}(\sum_{i=1}^j n_{i})(r - \sum_{i=1}^j n_{i})} L_{-n_k}\dots L_{-n_{2}}\nonumber\\
&\Bigg[ \left(-\frac{(r-n_1)\chi}{2}+\frac{(n_1-1)}{2\chi} \right)\frac{\alpha_{l_1}}{(t_{l_1}-t)^{n_1}} -\sum_{l_2:l_2\neq l_1}\sum_{i=1}^{n_1-1} \frac{\alpha_{l_1}\alpha_{l_2}}{4}\frac{1}{(t_{l_1}-t)^{i}(t_{l_2}-t)^{n_1-i}}\Bigg]f.
\end{align}

Certain terms with order $k$ and $k+1$ (here the order means the number of $n_i$) cancel among themselves, for example, for fixed $(n_1,\dots,n_{k})$ and $l_1$, we consider the following term with order $k+1$:
\begin{equation}\label{equation 3.7}
\sum_{n_1'+n_1''=n_1}  \frac{(\chi^2)^{r-k}}{\prod_{j=1}^{k-1}(\sum_{i=1}^j n_{i})(r - \sum_{i=1}^j n_{i})} \frac{1}{\chi^2 n_1'(r-n_1')} L_{-n_k}\dots L_{-n_2} L_{-n_1''}(-\frac{(r-n_1')\chi}{2} )\frac{\alpha_{l_1}}{(t_{l_1}-t)^{n_1'}} f.
\end{equation}
If we extract the term that depends on $t_{l_1}$ in the operator $L_{-n_1''}$, we have
\begin{align*}
\left(-\frac{1}{(t_{l_1}-t)^{n_1''-1}}\partial_{t_{l_1}}+\frac{Q(n_1''-1)\alpha_{l_1}}{2(t_{l_1}-t)^{n_1''}} \right)(-\frac{(r-n_1')\chi}{2} )\frac{\alpha_{l_1}}{(t_{l_1}-t)^{n_1'}} f=-\frac{n_1'(r-n_1')\chi}{2} \frac{\alpha_{l_1}}{(t_{l_1}-t)^{n_1}} f,
\end{align*}
In this equation we have eliminated all the terms that contain $\alpha_{l_1}^2$. For example, it is not hard to see that $\frac{\alpha_{l_1}}{(t_{l_1}-t)^{n_1'+n_1''-1}} \partial_{t_{l_1}} f = 0$. The previous calculus shows that if we extract the term $-\frac{1}{(t_{l_1}-t)^{n_1''-1}}\partial_{t_{l_1}}+\frac{Q(n_1''-1)\alpha_{l_1}}{2(t_{l_1}-t)^{n_1''}}$ from $L_{-n_1''}$ in \eqref{equation 3.7}, we can simplify as follows to obtain a term of order $k$:
\begin{align*}
&-\sum_{n_1'+n_1''=n_1}  \frac{(\chi^2)^{r-k}}{\prod_{j=1}^{k-1}(\sum_{i=1}^j n_{i})(r - \sum_{i=1}^j n_{i})}  L_{-n_k}\dots L_{-n_2} L_{-n_1''}\frac{1}{2\chi } \frac{\alpha_{l_1}}{(t_{l_1}-t)^{n_1'}} f\\
=& -\frac{(\chi^2)^{r-k}}{\prod_{j=1}^{k-1}(\sum_{i=1}^j n_{i})(r - \sum_{i=1}^j n_{i})} L_{-n_k}\dots L_{-n_2}\frac{(n_1-1)}{2\chi}\frac{\alpha_{l_1}}{(t_{l_1}-t)^{n_1}} f.
\end{align*}
The last line that we extract from \eqref{equation 3.7} cancels a term of order $k$ in \eqref{equation 3.6}. Thus, after such cancellations,
\begin{align}\label{equation 3.8}
\mathcal{D}_r f =& \sum_{l_1=1}^N \sum_{k=1}^r \sum_{n_1+\dots+n_k=r} \frac{(\chi^2)^{r-k}}{\prod_{j=1}^{k-1}(\sum_{i=1}^j n_{i})(r - \sum_{i=1}^j n_{i})} \Bigg[ (-\frac{(r-n_1)\chi}{2})L_{-n_k}\dots L^{\lg \cancel{l_1} \rg}_{-n_2}\frac{\alpha_{l_1}}{(t_{l_1}-t)^{n_1}}\nonumber\\
& -L_{-n_k}\dots L_{-n_2}\sum_{l_2:l_2 \neq l_1}\sum_{i=1}^{n_1-1} \frac{\alpha_{l_1}\alpha_{l_2}}{4}\frac{1}{(t_{l_1}-t)^{i}(t_{l_2}-t)^{n_1-i}}\Bigg]f
\end{align}
where $L^{\lg \cancel{l} \rg}_{-n} = \sum_{l'\neq l} \left(-\frac{\partial_{t_{l'}}}{(t_{l'}-t)^{n-1}}+\frac{Q(n-1)\alpha_{l'}}{2(t_{l'}-t)^n}\right)$.

Again by fixing $(n_1,\dots,n_{k})$, we consider the following term of order $k+1$:
\begin{small}
\begin{align*}
&\sum_{n_1'+n_1''=n_1} \frac{(\chi^2)^{r-k}}{\prod_{j=1}^{k-1}(\sum_{i=1}^j n_{i})(r - \sum_{i=1}^j n_{i})}  \frac{1}{\chi^2 n_1'(r-n_1')}(-\frac{(r-n_1')\chi}{2}) L_{-n_k}\dots L_{-n_2}L^{\lg \cancel{l_1} \rg}_{-n_1''}\frac{\alpha_{l_1}}{(t_{l_1}-t)^{n_1'}}f \\
=& -\sum_{l_2:l_2\neq l_1}\sum_{n_1'+n_1''=n_1} \frac{(\chi^2)^{r-k}}{\prod_{j=1}^{k-1}(\sum_{i=1}^j n_{i})(r - \sum_{i=1}^j n_{i})}\frac{1}{2\chi n_1'} L_{-n_k}\dots L_{-n_2}\\
&\Bigg[(-\frac{(r-n_1'')\chi}{2}+\frac{(n_1''-1)}{2\chi})\frac{\alpha_{l_1}\alpha_{l_2}}{(t_{l_1}-t)^{n_1'}(t_{l_2}-t)^{n_1''}} -\sum_{l_3:l_3 \notin \{l_1,l_2\}}\sum_{i=1}^{n_1''-1} \frac{\alpha_{l_1}\alpha_{l_2} \alpha_{l_3}}{4}\frac{1}{(t_{l_1}-t)^{n_1'}(t_{l_2}-t)^{i}(t_{l_3}-t)^{n_1''-i}}\Bigg]f,
\end{align*}
\end{small}
where in the equality we interchange $L^{\lg \cancel{l_1} \rg}_{-n_1''}$ and $\frac{\alpha_{l_1}}{(t_{l_1}-t)^{n_1'}}$ and then apply \eqref{equation Lf} to calculate $L^{\lg \cancel{l_1} \rg}_{-n_1''} f$. The terms with $\alpha_{l_1}^2$ or $\alpha_{l_2}^2$ were eliminated.
Then we get
\begin{align}\label{equation 3.9}
\mathcal{D}_r f =&\sum_{l_1\neq l_2} \sum_{k=2}^r \sum_{n_1+\dots+n_k=r} \frac{(\chi^2)^{r-k}}{\prod_{j=1}^{k-1}(\sum_{i=1}^j n_{i})(r - \sum_{i=1}^j n_{i})} (-\frac{(r-n_1)\chi}{2})L_{-n_k}\dots L_{-n_3}\nonumber\\
&\Bigg[(-\frac{(r-n_{1}-n_{2})\chi}{2}+\frac{(n_{2}-1)}{2\chi})\frac{\alpha_{l_1}\alpha_{l_2}}{(t_{l_1}-t)^{n_1}(t_{l_2}-t)^{n_2}} \nonumber\\
&-\sum_{l_3:l_3 \notin \{l_1,l_2\}}\sum_{i=1}^{n_2-1} \frac{\alpha_{l_1}\alpha_{l_2} \alpha_{l_3}}{4}\frac{1}{(t_{l_1}-t)^{n_1}(t_{l_2}-t)^{i}(t_{l_3}-t)^{n_2-i}}\Bigg]f
\end{align}
Comparing the expression above to the expression \eqref{equation 3.6}, we can proceed with a recurrence. Suppose that we arrive at the following expression, with $1\le K \le r-1$:
\begin{align}\label{A12}
\mathcal{D}_r f =&\sum_{\substack{l_1,\dots,l_K: \\ \forall i\neq j, l_i \neq l_j}} \sum_{k=K}^r \sum_{n_1+\dots+n_k=r} \frac{(\chi^2)^{r-k}}{\prod_{j=1}^{k-1}(\sum_{i=1}^j n_{i})(r - \sum_{i=1}^j n_{i})} \prod_{j=1}^{K-1}(-\frac{(r-\sum_{i=1}^j n_{i})\chi}{2}) L_{-n_k}\dots L_{-n_{K+1}} \nonumber\\
& \Bigg[(-\frac{(r-\sum_{i=1}^{K} n_{i})\chi}{2}+\frac{(n_{K}-1)}{2\chi})\prod_{j=1}^K \frac{\alpha_{l_j}}{(t_{l_j}-t)^{n_{j}}} \nonumber\\
&-\sum_{l_{K+1}:l_{K+1} \notin \{l_1,\dots,l_K\}}\sum_{i=1}^{n_K-1} \left(\prod_{j=1}^{K-1} \frac{\alpha_{l_j}}{(t_{l_j}-t)^{n_{j}}} \right)  \frac{\alpha_{l_K} \alpha_{l_{K+1}}}{4(t_{l_K}-t)^{i}(t_{l_{K+1}}-t)^{n_K-i}}\Bigg]f.
\end{align}

We will repeat what we have done from \eqref{equation 3.6} to \eqref{equation 3.9} in this general setting. For fixed tuple of $(n_1,\dots,n_k, l_1, \dots , l_K)$, consider the following configurations with $1 \le L \le K $  \[(n_1,\dots, n_{L-1},n_{K}',n_{L},\dots,n_{K-1},n_{K}'',n_{K+1},\dots,n_k)\]
Remark that if $L=K$, the decomposition is simply $(n_1,\dots,n_{K-1},n_K',n_{K}'',n_{K+1},\dots,n_k)$. The new configuration has $k+1$ terms and we can find some similar cancellations as previously by investigating the following term of order $k+1$:
\begin{small}
\begin{align*}
&\sum_{L=1}^{K}\sum_{n_K'+n_K''=n_K} \frac{(\chi^2)^{r-k}}{\prod_{j=1}^{k-1}(\sum_{i=1}^j n_{i})(r - \sum_{i=1}^j n_{i})} \frac{1}{\chi^2}\frac{\prod_{j=L}^{K-1}(\sum_{i=1}^j n_{i})(r - \sum_{i=1}^j n_{i})}{\prod_{j=L-1}^{K-1}(\sum_{i=1}^j n_{i}+n_K')(r - \sum_{i=1}^j n_{i}-n_K')}\\
&\quad \prod_{j=1}^{L-1}(-\frac{(r-\sum_{i=1}^j n_{i})\chi}{2})\prod_{j=L-1}^{K-2}(-\frac{(r-\sum_{i=1}^{j} n_{i}-n_K')\chi}{2})
L_{-n_k}\dots L_{-n_{K+1}}\\
&\quad (-\frac{(r-\sum_{i=1}^{K-1} n_{i}-n_{K}')\chi}{2})\left(-\frac{\partial_{t_{l_L}}}{(t_{l_L}-t)^{n_K''-1}}+\frac{\alpha_{l_L}}{(t_{l_L}-t)^{n_K''}} \right)
\sum_{\substack{l_1,\dots,l_K: \\ \forall i\neq j, l_i \neq l_j}}  \left(\prod_{\substack{j=1\\j \neq L}}^{K-1} \frac{\alpha_{l_j}}{(t_{l_j}-t)^{n_{j}}}\right) \frac{\alpha_{l_L}}{(t_{l_L}-t)^{n_{K}'}}f\end{align*}
\end{small}
Remark that this corresponds to the term extracted from \eqref{equation 3.7}. After some simplifications, this equals
\begin{align*}
& -\sum_{n_K'+n_K''=n_K}  \frac{(\chi^2)^{r-k}}{\prod_{j=1}^{k-1}(\sum_{i=1}^j n_{i})\prod_{j=K}^{k-1}(r - \sum_{i=1}^j n_{i})}\frac{1}{2\chi}(-\frac{\chi}{2})^{K-1}\left( \sum_{L=1}^{K}\frac{n_K'\prod_{j=L}^{K-1}(\sum_{i=1}^j n_{i})}{\prod_{j=L-1}^{K-1}(\sum_{i=1}^j n_{i}+n_K')} \right)   \\
&\qquad L_{-n_k}\dots L_{-n_{K+1}}\sum_{\substack{l_1,\dots,l_K: \\ \forall i\neq j, l_i \neq l_j}} \left( \prod_{j=1}^{K} \frac{\alpha_{l_j}}{(t_{l_j}-t)^{n_{j}}} \right) f.
\end{align*}
We claim that 
\begin{equation}
 \sum_{L=1}^{K}\frac{n_K'\prod_{j=L}^{K-1}(\sum_{i=1}^j n_{i})}{\prod_{j=L-1}^{K-1}(\sum_{i=1}^j n_{i}+n_K')}=1,
\end{equation}
which is not difficult to prove by induction on $K$: the case $K=1$ is trivial. Suppose that the identity holds for $K=K_0$, by induction we consider $K=K_0+1$:
\begin{equation}
\sum_{L=1}^{K}\frac{n_K'\prod_{j=L}^{K-1}(\sum_{i=1}^j n_{i})}{\prod_{j=L-1}^{K-1}(\sum_{i=1}^j n_{i}+n_K')}=\frac{\sum_{i=1}^{K_0} n_{i}}{\sum_{i=1}^{K_0} n_{i}+n_{K_0+1}'}+\frac{n_{K_0+1}'}{\sum_{i=1}^{K_0} n_{i}+n_{K_0+1}'}=1,
\end{equation}
where we separate the sum in to $\sum_{L=1}^{K_0}+\sum_{L=K_0+1}$. Therefore we can further simplify the expression above for the term of order $k+1$, which equals now
\begin{small}
\begin{align*}
-\sum_{\substack{l_1,\dots,l_K: \\ \forall i\neq j, l_i \neq l_j}} \frac{(\chi^2)^{r-k}}{\prod_{j=1}^{k-1}(\sum_{i=1}^j n_{i})(r - \sum_{i=1}^j n_{i})} \prod_{j=1}^{K-1}(-\frac{(r-\sum_{i=1}^j n_{i})\chi}{2})L_{-n_k}\dots L_{-n_{K+1}}  \frac{(n_{K}-1)}{2 \chi} \left(\prod_{j=1}^K \frac{\alpha_{l_j}}{(t_{l_j}-t)^{n_{j}}}\right) f.
\end{align*}
\end{small}
The last line cancels a term of order $k$ in \eqref{A12}. Therefore, after canceling the terms, we arrive at an expression which is the analogue of \eqref{equation 3.8}
\begin{align}
\mathcal{D}_r f =&\sum_{\substack{l_1,\dots,l_K: \\ \forall i\neq j, l_i \neq l_j}} \sum_{k=K}^r \sum_{n_1+\dots+n_k=r} \frac{(\chi^2)^{r-k}}{\prod_{j=1}^{k-1}(\sum_{i=1}^j n_{i})(r - \sum_{i=1}^j n_{i})} \prod_{j=1}^{K-1}(-\frac{(r-\sum_{i=1}^j n_{i})\chi}{2})\nonumber\\
& \Bigg[L_{-n_k}\dots L_{-n_{K+1}}^{\lg \cancel{l_1},\dots,\cancel{l_K} \rg}(-\frac{(r-\sum_{i=1}^{K} n_{i})\chi}{2}) \left(\prod_{j=1}^K \frac{\alpha_{l_j}}{(t_{l_j}-t)^{n_{j}}}\right) -L_{-n_k}\dots L_{-n_{K+1}}\nonumber\\
&\sum_{l_{K+1}:l_{K+1} \notin \{l_1,\dots,l_K\}}\sum_{i=1}^{n_K-1} \left(\prod_{j=1}^{K-1} \frac{\alpha_{l_j}}{(t_{l_j}-t)^{n_{j}}} \right) \frac{\alpha_{l_K} \alpha_{l_{K+1}}}{4(t_{l_K}-t)^{i}(t_{l_{K+1}}-t)^{n_K-i}}\Bigg] f
\end{align}
The last step is to develop the operator $L_{-n_{K+1}}^{\lg \cancel{l_1},\dots,\cancel{l_K} \rg}$ exactly as what we did to \eqref{equation 3.8} and we will obtain
\begin{align}
\mathcal{D}_r f =&\sum_{\substack{l_1,\dots,l_{K+1}: \\ \forall i\neq j, l_i \neq l_j}} \sum_{k={K+1}}^r \sum_{n_1+\dots+n_k=r} \frac{(\chi^2)^{r-k}}{\prod_{j=1}^{k-1}(\sum_{i=1}^j n_{i})(r - \sum_{i=1}^j n_{i})} \prod_{j=1}^{K}(-\frac{(r-\sum_{i=1}^j n_{i})\chi}{2}) \nonumber\\
& L_{-n_k}\dots L_{-n_{K+2}}\Bigg[(-\frac{(r-\sum_{i=1}^{K+1} n_{i})\chi}{2}+\frac{(n_{K+1}-1)}{2 \chi}) \left(\prod_{j=1}^K \frac{\alpha_{l_j}}{(t_{l_j}-t)^{n_{j}}}\right) \nonumber\\
&-\sum_{l_{K+2}:l_{K+2} \notin \{l_1,\dots,l_{K+1}\}}\sum_{i=1}^{n_{K+1}-1} \left( \prod_{j=1}^{K} \frac{\alpha_{l_j}}{(t_{l_j}-t)^{n_{j}}} \right) \frac{\alpha_{l_{K+1}} \alpha_{l_{K+2}}}{4(t_{l_{K+1}}-t)^{i}(t_{l_{K+2}}-t)^{n_{K+1}-i}}\Bigg]f
\end{align}
This allows us to go from $K$ to $K+1$ in the statement \eqref{A12}. When $K$ grows to $r$, as $n_i=1$ for all $1\le i \le r$ and $\sum_{i=1}^{r} n_{i}=r$, we obtain
\begin{equation}
\mathcal{D}_r f = 0
\end{equation}
under the condition \eqref{equation quotient}. By discussions at the beginning of the proof, this allows to conclude that $\lambda_{\mathbf{n}}(\gamma)=0$ for all $|\mathbf{n}|=r$, hence $\mathcal{D}_r f = 0$ in the general setting.
\end{proof}

\appendix

\section{Proof of the derivative rule}\label{section proof derivative}

Let us first recall the Gaussian integration by parts formula:

\begin{lem}[Gaussian integration by parts]\label{ipp gaussien discrete}
Let $(Y_1,Y_2)\in \mathbb{R}\times \mathbb{R}^d, d\ge 1$ a centered Gaussian vector, and $\phi \in \mathcal{C}^{\infty}(\mathbb{R}^d)$ a function that decays faster than any polynomials at infinity. Then,
	\begin{equation}
	\E[Y_1 \phi(Y_2)]=\E[Y_1Y_2]\E[\nabla \phi(Y_2)]
	\end{equation}
\end{lem}
\begin{proof}
Let $(\lambda,\mu) \in \mathbb{R}\times \mathbb{R}^d$, we calculate
	\begin{equation*}
	\E[e^{\lambda Y_1+\mu\cdot Y_2}]=e^{\frac{\lambda^2}{2}\E[Y_1^2]+\lambda \mu \cdot \E[Y_1Y_2]+\frac{1}{2}{}^t\mu \textnormal{Var}(Y_2) \mu}
	\end{equation*}
Taking the derivative with respect to $\lambda$ and evaluate at $\lambda =0$, we obtain:
\[\E[Y_1e^{\mu\cdot Y_2}]=\mu \cdot \E[Y_1Y_2] e^{\frac{1}{2}{}^t\mu \textnormal{Var}(Y_2) \mu}\]
This proves the forumla for the function $\phi(y)=e^{\mu \cdot y}$. We then conclude with an argument of density.
\end{proof}

To calculate derivatives of the correlation functions, we will need a "continuous" version of Gaussian integration by parts, where $Y_2$ is now of infinite dimension:

\begin{lem}
Let $\epsilon > 0$ and $f$ a smooth test function with compact support. Denote $X(f) = (X,f)$, we have 
	\begin{equation}
	\begin{split}
	\lg X(f) \prod_{l=0}^N V_{\alpha_l,\epsilon}(z_l)\rg_{\delta} =& \sum_{l=0}^N \alpha_l \E[X(f) X_{\epsilon}(z_l)]\lg \prod_{l=0}^N V_{\alpha_l,\epsilon}(z_l)\rg_{\delta}\\
	&-\mu \gamma \intc \theta_{\delta}(y-z_0 )\E[X(f) X_{\epsilon}(y)]\lg V_{\gamma,\epsilon}(y) \prod_{l=0}^N V_{\alpha_l,\epsilon}(z_l)\rg_{\delta} d^2y
	\end{split}
	\end{equation}
\end{lem}

This result can be obtained from the previous Gaussian integration by parts formula by discretizing the Gaussian multiplicative chaos measure. Consider a function $f$ such that $\intc f = 0$. By definition, we have for $\epsilon > 0$ and $x \in \mathbb{C}$,
\begin{equation*}
 \E[X(f) X_{\epsilon}(x)] = f*C_{\epsilon}(x) -\frac{1}{4}\intc \ln \hg(y)f(y)d^2y
\end{equation*}
where $C_{\epsilon}=\ln \frac{1}{|z|}*\eta_{\epsilon}$. Recall that $\eta_{\epsilon}=\frac{1}{\epsilon^2}\eta(\frac{|x|^2}{\epsilon^2})$ with $\eta$ supported in $[\frac{1}{2},1]$. Then by the Gaussian integration by parts formula,
\begin{align*}
&\lg (X + \frac{Q}{2} \ln \hg)(f)\prod_{l=0}^N V_{\alpha_l,\epsilon}(z_l)\rg_{\delta} \\
=& \sum_{l=0}^N \alpha_l \,f*C_{\epsilon}(z_l)\lg \prod_{l=0}^N V_{\alpha_l,\epsilon}(z_l)\rg_{\delta}-\mu \gamma \intc\theta_{\delta}(y-z_0 )  f*C_{\epsilon}(y) \lg V_{\gamma,\epsilon}(y)\prod_{l=0}^N V_{\alpha_l,\epsilon}(z_l)\rg_{\delta} d^2y\\
&+\frac{1}{4}\intc  \ln \hg(y)f(y)d^2y \Big( (2Q-\sum_{l=0}^N \alpha_l)\lg \prod_{l=0}^N V_{\alpha_l,\epsilon}(z_l)\rg_{\delta}+\mu \gamma \intc \theta_{\delta}(y'-z_0 )\lg V_{\gamma,\epsilon}(y')\prod_{l=0}^N V_{\alpha_l,\epsilon}(z_l)\rg_{\delta}d^2y'\Big)\\
=& \sum_{l=0}^N \alpha_l \,f*C_{\epsilon}(z_l)\lg \prod_{l=0}^N V_{\alpha_l,\epsilon}(z_l)\rg_{\delta}-\mu \gamma \intc\theta_{\delta}(y-z_0 )  f*C_{\epsilon}(y) \lg V_{\gamma,\epsilon}(y)\prod_{l=0}^N V_{\alpha_l,\epsilon}(z_l)\rg_{\delta} d^2y,
\end{align*}
where in the last equality we used the KPZ identity \eqref{equation KPZ} to help the cancellation. Now we take $f= \partial_z\eta_{\epsilon}(z_i-\cdot)$. Note that for $1\le i \le N$,
\begin{align*}
\partial_{z_i} \left(g_{\epsilon}(z_i)^{\Delta_{\alpha_i}} e^{\alpha_i X_{\epsilon}(z_i) - \frac{\alpha_i^2}{2} \E[X_{\epsilon}(z_i)^2]} \right) &= \alpha_i \partial_{z_i}(X_{\epsilon}(z_i) + \frac{Q}{2} \ln \hg_{\epsilon}(z_i))\, g_{\epsilon}(z_i)^{\Delta_{\alpha_i}} e^{\alpha_i X_{\epsilon}(z_i) - \frac{\alpha_i^2}{2} \E[X_{\epsilon}(z_i)^2]}\\
&=\alpha_i (X + \frac{Q}{2} \ln \hg)(f)\, g_{\epsilon}(z_i)^{\Delta_{\alpha_i}} e^{\alpha_i X_{\epsilon}(z_i) - \frac{\alpha_i^2}{2} \E[X_{\epsilon}(z_i)^2]}.
\end{align*}
Hence for $1\le i \le N$,
\begin{align*}
\partial_{z_i}\lg \prod_{l=0}^N V_{\alpha_l,\epsilon}(z_l)\rg_{\delta} = &\sum_{\substack{j=0\\j\neq i}}^N  \frac{\alpha_i\alpha_j}{2(z_j-z_i)_{\epsilon}}\lg \prod_{l=1}^N V_{\alpha_l,\epsilon}(z_l)\rg_{\delta}-\frac{\mu \gamma \alpha_i}{2} \intc \frac{\theta_{\delta}(y-z_0)}{(y-z_i)_{\epsilon}}\lg V_{\gamma,\epsilon}(y)\prod_{l=0}^N V_{\alpha_l,\epsilon}(z_l)\rg_{\delta}d^2y.
\end{align*}
When $i = 0$, there is an additional term coming from the derivative of the regularization $\theta_{\delta}(y-z_0)$, which gives
\begin{align*}
\mu \intc \partial_z \theta_{\delta}(y-z_0)\lg V_{\gamma,\epsilon}(y)\prod_{l=0}^N V_{\alpha_l,\epsilon}(z_l)\rg_{\delta} d^2y,
\end{align*}
this concludes the proof for the derivative rule \ref{lem derivation G}.

\end{document}